\documentclass[final,twoside,11pt]{entics} 
\usepackage{enticsmacro}
\usepackage{graphicx}
\usepackage{enumerate}
\usepackage[all]{xy}
\usepackage{amsmath,amssymb,wasysym}
\usepackage{url}
\usepackage{xcolor}
\usepackage{tikz-cd}
\usepackage{cmll}
\usepackage{ebproof}
\def\today{January 31, 2024}
\def\conf{MFPS 2023} 	
\def\myconf{\conf}
\volume{3}			
\def\volu{3}			
\def\papno{14}			
\def\lastname{Lemay, Vienney} 
\def\runtitle{Graded Differential Categories} 
\def\mydoi{12290}
\def\copynames{J.-S. P. Lemay, J.-B. Vienney}    
\begin{document}
\begin{frontmatter}
  \title{Graded Differential Categories and\\ Graded Differential Linear Logic\thanksref{ALL}} 	
  \subtitle{\textsc{Corrigendum}}					
 \thanks[ALL]{For this research, the first author was financially supported by a JSPS Postdoctoral Fellowship, Award \#: P21746, while the second author was financially supported by the University of Ottawa and Aix-Marseille University.}   
  \author{Jean-Simon Pacaud Lemay\thanksref{a}\thanksref{myemail}}	
   \author{Jean-Baptiste Vienney\thanksref{b}\thanksref{coemail}}		
   \address[a]{School of Mathematical and Physical Sciences
 \\ Macquarie University\\				
    Sydney, Australia}  							
   \thanks[myemail]{Email: \href{mailto:js.lemay@mq.edu.au} {\texttt{\normalshape
        js.lemay@mq.edu.au}}} 
  \address[b]{Department of Mathematics and Statistics\\University of Ottawa\\
    Ottawa, Canada} 
  \thanks[coemail]{Email:  \href{mailto:jvien065@uottawa.ca} {\texttt{\normalshape
        jeanbaptiste.vienney@gmail.com}}}
\begin{abstract} In Linear Logic ($\mathsf{LL}$), the exponential modality $\oc$ brings forth a distinction between non-linear proofs and linear proofs, where linear means using an argument exactly once. Differential Linear Logic ($\mathsf{DiLL}$) is an extension of Linear Logic which includes additional rules for $\oc$ which encode differentiation and the ability of linearizing proofs. On the other hand, Graded Linear Logic ($\mathsf{GLL}$) is a variation of Linear Logic in such a way that $\oc$ is now indexed over a semiring $R$. This $R$-grading allows for non-linear proofs of degree $r \in R$, such that the linear proofs are of degree $1 \in R$. There has been recent interest in combining these two variations of $\mathsf{LL}$ together and developing Graded Differential Linear Logic ($\mathsf{GDiLL}$). In this paper we present a sequent calculus for $\mathsf{GDiLL}$, as well as introduce its categorical semantics, which we call graded differential categories, using both coderelictions and deriving transformations. We prove that symmetric powers always give graded differential categories, and provide other examples of graded differential categories. We also discuss graded versions of (monoidal) coalgebra modalities, additive bialgebra modalities, and the Seely isomorphisms, as well as their implementations in the sequent calculus of $\mathsf{GDiLL}$. 
\end{abstract}
\end{frontmatter}
\section{Introduction}
\allowdisplaybreaks

Linear Logic ($\mathsf{LL}$) was introduced by Girard in \cite{girard_linear_1987}, and is a resource-sensitive logic whose categorical semantics are based in monoidal category theory \cite{mellies_categorical_2008}. A key feature of $\mathsf{LL}$ is that there is a natural notion of linearity, allowing for a distinction between non-linear proofs and linear proofs. In many models of $\mathsf{LL}$, the logical notion of linear, meaning using an argument exactly once, coincides with the algebraic notion of linear, meaning preserving addition. In the multiplicative and exponential fragments of $\mathsf{LL}$ ($\mathsf{MELL}$), the separation between linear and non-linear is captured thanks to the exponential modality, which is a unary connective $\oc$, read as ``of course'' or ``bang''. A non-linear proof of $A$ implies $B$ is of the form $\oc A \vdash B$, while a linear proof is of the form $A \vdash B$. A categorical model of $\mathsf{MELL}$, often called a linear category \cite[Def 3]{bierman1995categorical}, is a symmetric monoidal closed category equipped with monoidal coalgebra modality \cite[Sec 3]{blute2015cartesian}, also sometimes called a linear exponential modality, which in particular is a comonad $\oc$ such that each $\oc A$ is naturally a cocommutative comonoid. So categorically speaking, a non-linear map from $A$ to $B$ is a map of type $\oc A \to B$, while a linear map from $A$ to $B$ is one of type $A \to B$. One of the great aspects of monoidal category theory is its flexibility. As such, this allows for many natural variations, extensions, and generalizations of $\mathsf{LL}$ and its fragments. 

In one direction there is Differential Linear Logic ($\mathsf{DiLL}$), which was introduced by Ehrhard and Regnier in \cite{ehrhard2006differential}. In particular, $\mathsf{DiLL}$ is an extension of $\mathsf{MELL}$ which adds a differentiation structural rule for $\oc$, which allows one to differentiate non-linear proofs to obtain linear proofs. The categorical semantics of $\mathsf{DiLL}$ is provided by differential categories, which were introduced by Blute, Cockett, and Seely in \cite{blute2006differential}. Differential categories were further developed by Fiore \cite{fiore2007differential} and Ehrhard \cite{ehrhard2017introduction}, and then again revisited by Blute, Cockett, Seely, and the first author in \cite{Blute2019}. In a differential category, the non-linear maps $\oc A \to B$ are interpreted as smooth maps, meaning infinitely differentiable. The differential structure can be defined in terms of either a deriving transformation $\partial: \oc A \otimes A \to \oc A$ \cite[Def 2.5]{blute2006differential} or a codereliction $\overline{\mathsf{d}}: A \to \oc A$ \cite[Def 4.11]{blute2006differential}, which were shown to indeed be equivalent in \cite{Blute2019,cockett2017there}. 

In another direction there is Graded Linear Logic ($\mathsf{GLL}$), whose conception originated from Girard, Scedrov, and Scott's bounded linear logic \cite{girard1992bounded}. The exponential fragment of $\mathsf{GLL}$ is now graded over a fixed semiring $R$, which means that there are unary connectors $\oc_r$ for each $r \in R$. This allows for non-linear proofs of degree $r$, $\oc_r A \vdash B$, such that the linear proofs are of degree $1 \in R$. Categorically speaking, the $R$-graded exponential modality is captured by an $R$-graded monoidal coalgebra modality \cite[Def 1]{katsumata2018double}, also called a graded linear exponential comonad, first introduced by Brunel et al. in \cite{brunel2014core} (under the name exponential action), and then further studied by Breuvart and Pagani in \cite{breuvart2015modelling}, Gaboardi et al. in \cite{gaboardi2016combining}, and Katsumata in \cite{katsumata2018double}. In particular, an $R$-graded monoidal coalgebra modality is an $R$-graded comonad such that all the $\oc_r A$ naturally form a graded comonoid. So a non-linear map of degree $r$ is a map of type $\oc_r A \to B$, and from a linear map $A \to B$, we can obtain a non-linear map of degree $1$, $\oc_1 A \to B$. 

There has been a recent desire to combine $\mathsf{DiLL}$ and $\mathsf{GLL}$ \cite{breuvart2023unifying,kerjean2023taylor} to obtain Graded Differential Linear Logic ($\mathsf{GDiLL}$). As such, in this paper, we introduce a sequent calculus for $\mathsf{GDiLL}$, and also introduce graded differential categories (Sec \ref{sec:gdiffcat}), which provide the categorical semantics of $\mathsf{GDiLL}$. We provide graded versions of both the deriving transformation, now of type $\partial_r: \oc_r A \otimes A \to \oc_{r+1} A$ (Def \ref{def:diffcat}), and the codereliction, now of type $\overline{\mathsf{d}}: A \to \oc_1 A$ (Def \ref{def:coder}). We will also explain why deriving transformations and coderelicitons are still equivalent in the graded setting (Thm \ref{thm:der-coder}). The main intuition in a graded differential category is that the derivative of a smooth map of degree $r+1$ has degree $r$, which is analogous to the fact that the derivative of $x^{n+1}$ is $(n+1) \cdot x^n$, and that the linear maps are precisely the smooth maps of degree $1$. This intuition is nicely captured by symmetric powers, which we show always give a graded differential category and capture differentiating (homogenous) polynomials (Sec \ref{sec:sympowers}). We also discuss other examples of graded differential categories, including $R$-valued multisets. 

Graded differential categories also provide a solution to the problem presented by the first author in \cite{lemay2020fhilb}. Indeed, in said paper, it is explained why the category of finite dimensional Hilbert spaces ($\mathsf{FHILB}$), the main model of interest in categorical quantum mechanics, has no non-trivial differential category structure. Briefly, in the non-graded setting, $\oc A$ has a distinct ``infinite dimensional'' flavour to it. However, in the graded setting, $\oc_r A$ no longer carries this ``infinite dimensional'' flavour. In particular, this is the case regarding the symmetric powers construction, since individually the symmetric powers of a finite-dimensional Hilbert space are again finite-dimensional. Therefore, while $\mathsf{FHILB}$ may not be an interesting (co)differential category, $\mathsf{FHILB}$ is in fact a very interesting graded (co)differential category. As such, graded differential categories allow us to consider more ``finite dimensional'' flavoured models of differential categories and provide a path for studying differential category theory in categorical quantum mechanics. 

In order to properly define graded differential categories, we also review graded versions of coalgebra modalities and monoidal coalgebra modalities (Sec \ref{sec:gmod}), as well as introduce the novel concept of a graded additive bialgebra modality (Sec \ref{sec:gbialg}). This latter concept provides a solution to the problem regarding Seely isomorphisms in the graded setting. Indeed, in the non-graded setting, a monoidal coalgebra modality can equivalently be described as either a storage modality \cite[Def 10]{Blute2019}, which in particular is a comonad $\oc$ with the Seely isomorphisms $\oc(A \times B) \cong \oc A \otimes \oc B$ (where $\otimes$ is the monoidal product and $\times$ is the product), or as an additive bialgebra modality \cite[Def 5]{Blute2019}, which in particular is a comonad $\oc$ where $\oc A$ is naturally a bialgebra. Unfortunately in the graded setting, the $R$-grading makes the constructions between these concepts no longer properly typed. In particular for the Seely isomorphisms, one should not expect that $\oc_r(A \times B)$ be isomorphic to $\oc_r A \otimes \oc_r B$ -- in particular the symmetric powers example fails this. So we must be careful in the definitions of graded Seely isomorphisms and graded additive bialgebra modalities. 

The fact that, in the non-graded setting, monoidal coalgebra modalities and additive bialgebra modalities are indeed equivalent \cite[Thm 1]{Blute2019} was crucial for proving that there was a bijective correspondence between deriving transformations and coderelictions. Since we wish this bijective correspondence to still hold in the graded setting, we will require by definition that a graded additive bialgebra modality (Def \ref{def:addbialg}) be a graded monoidal coalgebra modality -- which was not required in the non-graded case. As such, graded additive bialgebra modalities lead us to the correct version of the graded Seely isomorphisms (Thm \ref{thm:Seely}) which is that $\oc_r(A \oplus B) \cong \bigoplus_{s+t =r} \oc_s A \otimes \oc_t B$ (where $\oplus$ is the biproduct since we will be working in a setting with additive structure). When taking the grading over the zero semiring $\mathsf{0}$, we recover the non-graded versions, that is, $\mathsf{0}$-graded additive bialgebra modalities, the $\mathsf{0}$-graded Seely isomorphism, and $\mathsf{0}$-graded differential categories are the same as additive bialgebra modalities, the Seely isomorphism, and differential categories. 

\textbf{Corrigendum:} The original version of this paper incorrectly asserted that the running example worked over any semiring $R$. However, the statement of Ex \ref{ex:RELsym} is correct, so our example does work when $R = \mathbb{N}$. This Corrigendum corrects the error by fixing $R = \mathbb{N}$ in Ex \ref{ex:REL!}, \ref{ex:RELd}, \ref{ex:RELbialg}, \ref{ex:RELaddbialg}, and \ref{ex:RELcoder}. This does not affect the overall story or theory, and the rest of the paper remains unchanged. 

\textbf{Conventions:} The graded story involves heavy usage of indexing. To not overload notation, we will omit the object index for identity maps and natural transformations. In an arbitrary category, we write identity maps as $1: A \to A$ and write composition diagrammatically, that is, the composition of maps $f: A \to B$ and $g: B \to C$ is denoted $f;g: A \to C$, which first does $f$ then $g$. For the definitions of this paper, we provide the necessary axioms as equations, while commutative diagram versions for these equations can be found in \cite[App A]{arxiv2023version}, which is a version of this paper which includes an appendix. 

\section{Graded Linear Logic and Graded Coalgebra Modalities}\label{sec:gmod}

In this section, we review the sequent calculus for the exponential fragment of $\mathsf{GLL}$ and its categorical semantics, specifically graded (monoidal) coalgebra modalities. In previous work \cite{breuvart2015modelling,brunel2014core,katsumata2018double}, the grading of $\mathsf{GLL}$ is often considered to be over a partially ordered semiring. For the story of this paper, we will slightly generalize and consider a grading over just a semiring. So throughout this paper, an arbitrary semiring is denoted as the tuple $R = ( \vert R \vert, +, \ast, 0, 1)$, where $\vert R \vert$ is the underlying set, $+$ is the addition, $\ast$ is the multiplication, $0$ the additive unit, and $1$ is the multiplicative unit. As a shorthand, when there is no confusion, we will simply write the multiplication as $rs := r \ast s$. 

In this paper, we work in the intuitionistic setting, so we do not assume negations (since it does not play a role in the differential story). The multiplicative and additive fragments of $\mathsf{GLL}$ are the same as $\mathsf{LL}$ (see \cite{mellies_categorical_2008} for a review of these fragments). For the exponential fragment of $\mathsf{GLL}$, we now have for each $r \in R$ a unary connector $\oc_r$, called the exponential modality at $r$ or just read as "bang $r$", and which together satisfy graded versions of the four exponential modality rules from $\mathsf{LL}$. In this case however, the \textbf{weakening} rule ($\mathsf{w}$) is only for $\oc_0$, the \textbf{dereliction} rule ($\mathsf{d}$) is only for $\oc_1$, while the \textbf{contraction} rule ($\mathsf{c}$) depends on pairs of elements $r,s \in R$ and the \textbf{promotion} rule ($\mathsf{prom}$) depends on tuples of elements $r, s_i \in R$: 
\begin{align}
    \begin{prooftree}
\hypo{\oc_{s_{1}}A, ..., \oc_{s_{n}}A \vdash B}
\infer1[$\mathsf{prom}_{r, s_1, \hdots, s_n}$]{\oc_{r s_{1}}A, ..., \oc_{r s_{n}}A \vdash \oc_{r}B}
\end{prooftree}
&&
\begin{prooftree}
\hypo{\Gamma, A \vdash B}
\infer1[$\mathsf{d}$]{\Gamma, \oc_{1}A \vdash B}
\end{prooftree}
&&
\begin{prooftree}
\hypo{\Gamma \vdash B}
\infer1[$\mathsf{w}$]{\Gamma, \oc_{0}A \vdash B}
\end{prooftree}
&&
\begin{prooftree}
\hypo{\Gamma, \oc_{r}A, \oc_{s}A \vdash B}
\infer1[$\mathsf{c}_{r,s}$]{\Gamma, \oc_{r+s}A \vdash B}
\end{prooftree}
\end{align}
The four rules are each associated with a separate aspect of the semiring structure. The weakening is associated with the zero element, the dereliction to the multiplicative unit, the contraction to addition, and the digging to multiplication. The cut-elimination rules can be found in \cite[Fig 2]{breuvart2015modelling}. 

Now let's turn our attention to the underlying category theory. Recall that the categorical semantics of the multiplicative fragment of $\mathsf{LL}$ is interpreted by a symmetric monoidal closed category and the additive fragment by finite products (see \cite{mellies_categorical_2008} for a review of these details). Since neither being closed nor having products is necessary for the categorical definition of the exponential modality, we will not require them here (though we will add products to the story later in Sec \ref{sec:gbialg} when discussing the Seely isomorphisms). Thus we will for now simply work with symmetric monoidal categories. For simplicity, we allow ourselves to work within a symmetric \emph{strict} monoidal category, meaning that the associativity and unit isomorphisms for the monoidal product are identities. An arbitrary symmetric monoidal category is denoted as the tuple $\mathcal{L} = (\vert \mathcal{L} \vert, \otimes, I, \sigma)$ where $\vert \mathcal{L} \vert$ is the underlying category, $\otimes$ is the monoidal product, $I$ is the monoidal unit, and $\sigma: A \otimes B \to B \otimes A$ is the natural symmetry isomorphism. By strictness, we have that $A \otimes I = A = I \otimes A$ and can write $A_1 \otimes \hdots \otimes A_n$. 

In the categorical semantics of $\mathsf{GLL}$, the exponential modality is interpreted by an $R$-graded symmetric monoidal comonad such that altogether the $\oc_r A$ naturally form an $R$-graded cocommutative comonoid. Unpacking this, we have that $\oc_r$ is a functor, while the four exponential rules are given by natural transformations. The weakening is given by a natural transformation of type $\mathsf{w}: \oc_0 A \to I$ and the dereliction by a natural transformation of type $\mathsf{d}: \oc_1 A \to A$. On the other hand, the contraction is given by an $R$-indexed family of natural transformations $\mathsf{c}_{r,s}: \oc_{r+s} A \to \oc_r A \otimes \oc_s A$, while the promotion gives us three $R$-indexed families: (i) $\mathsf{p}_{r,s}: \oc_{rs} A \to \oc_r \oc_s A$, (ii) $\mu^\otimes_r: \oc_r A \otimes \oc_r B \to \oc_r(A \otimes B)$, and (iii) $\mu^I_r: I \to \oc_r I$. 

We will in fact split the definition in two, one without $\mu^\otimes$, and  $\mu^I$, and one with them. This is because we wish to provide a graded version of Blute, Cockett, and Seely's original definition of a differential category \cite{blute2006differential}, which does not require non-graded versions of $\mu^\otimes$ and $\mu^I$. So the following is the graded version of \cite[Def 1]{Blute2019}. 

\begin{definition}\label{def:coalg} An \textbf{$R$-graded coalgebra modality} on a symmetric monoidal category $\mathcal{L}$ is a tuple $(\oc, \mathsf{p}, \mathsf{d}, \mathsf{c}, \mathsf{w})$ consisting of a family of endofunctors $\oc$ where for every $r \in R$ we have an endofunctor ${\oc_r: \mathcal{L} \to \mathcal{L}}$; two families of natural transformations $\mathsf{p}$ and $\mathsf{c}$ where for every $r, s \in  R$ we have a natural transformation of type $\mathsf{p}_{r,s}: \oc_{rs} A \to \oc_r \oc_s A$ and $\mathsf{c}_{r,s}: \oc_{r+s} A \to \oc_r A \otimes \oc_s A$; two natural transformations $\mathsf{d}: \oc_1 A \to A$ and $\mathsf{w}: \oc_0 A \to I$, and such that the following equalities hold for all $r,s,t \in R$:
\begin{enumerate}[{\em (i)}]
\item $(\oc, \mathsf{p}, \mathsf{d})$ is an $R$-graded comonad:      
\begin{align}\label{comonadeq}
\mathsf{p}_{1,r}; \mathsf{d} = 1 = \mathsf{p}_{r,1}; \oc_r(\mathsf{d}) && \mathsf{p}_{r  s,t}; \mathsf{p}_{r,s} = \mathsf{p}_{r,s t}; \oc_r(\mathsf{p}_{s,t}) 
\end{align}
\item $(A, \mathsf{c}, \mathsf{w})$ is an $R$-graded cocommutative comonoid: 
\begin{align}\label{comonoideq}
   \mathsf{c}_{r+s,t}; (\mathsf{c}_{r,s} \otimes 1) = \mathsf{c}_{r,s+t}; (1 \otimes \mathsf{c}_{s+t}) && \mathsf{c}_{r,0}; (1 \otimes \mathsf{w}) = 1 = \mathsf{c}_{0,r}; (\mathsf{w} \otimes 1) && \mathsf{c}_{r,s}; \sigma = \mathsf{c}_{s,r}
\end{align}
\item $\mathsf{p}$ is an $R$-graded comonoid morphism: 
\begin{align}\label{deltacomonoideq} \mathsf{p}_{r+s,t}; \mathsf{c}_{r,s} = \mathsf{c}_{r  t,s  t}; (\mathsf{p}_{r,t} \otimes \mathsf{p}_{s,t}) && \mathsf{p}_{0,r}; \mathsf{w} = \mathsf{w}
\end{align}
\end{enumerate}
\end{definition} 

In a symmetric monoidal category with an $R$-graded coalgebra modality $\oc$, we can interpret both linear maps and non-linear maps of degree $r \in R$. A linear map from $A$ to $B$ is simply a map of type $A \to B$. On the other hand, a non-linear map of degree $r$ is a map of type $\oc_r A \to B$. We can use $\mathsf{d}$ to forget about the linearity of a linear map $f: A \to B$ and obtain a non-linear map of degree $1$ as the composite $\mathsf{d};f: \oc_1 A \to B$. We stress that in some models, not all non-linear maps of degree $1$ are linear maps. Using $\mathsf{p}$ we can compose non-linear maps $f: \oc_s A \to B$ and $g: \oc_r B \to C$ to obtain a non-linear map of degree $rs$ defined as the composite $\mathsf{p}_{r,s};\oc_r(f);g: \oc_{rs} A \to C$. We can take the product of non-linear maps $f: \oc_r A \to B$ and $g: \oc_s A \to C$ using $\mathsf{c}$ to obtain a non-linear map of degree $r+s$ defined as the composite $\mathsf{c}_{r,s};(f \otimes g): \oc_{r+s} A \to B \otimes C$. Lastly, a (categorical) point of $A$ is a map of type $I \to A$. Every point $b: I \to B$ induces a non-linear map of degree $0$ from any $A$ to $B$ by pre-composing by $\mathsf{w}$, that is, $\mathsf{w};b: \oc_0 A \to B$. If one wishes to write down a sequent calculus for coalgebra modalities, one would replace the promotion rule with a weaker version:
\begin{align}
        \begin{prooftree}
\hypo{\oc_{s}A \vdash B}
\infer1[$\mathsf{prom}_{r,s}$]{\oc_{rs}A \vdash \oc_r B}
\end{prooftree}
\end{align}
This weaker version still allows us to obtain $\mathsf{p}$ and the functoriality of $\oc_r$. 

Adding back $\mu^\otimes_r$, and  $\mu^I_r$, we obtain the notion of an $R$-graded monoidal coalgebra modality. In this paper, we've elected to use the term graded monoidal coalgebra modality rather than graded linear exponential comonad to more closely follow the terminology used in differential category literature, specifically \cite{Blute2019}. So the following is indeed the same as \cite[Def 1]{katsumata2018double} (but without the assumption of order), which is a graded version of \cite[Def 2]{Blute2019}. 

\begin{definition}\label{def:moncoalg} An \textbf{$R$-graded monoidal coalgebra modality} on a symmetric monoidal category $\mathcal{L}$ is a tuple $(\oc, \mathsf{p}, \mathsf{d}, \mathsf{c}, \mathsf{w}, \mu^\otimes, \mu^I)$ consisting of an $R$-graded coalgebra modality $(\oc, \mathsf{p}, \mathsf{d}, \mathsf{c}, \mathsf{w})$, and two families of natural transformations $\mu^\otimes$ and $\mu^I$ where for every $r \in R$ we have natural transformations of type $\mu^\otimes_r: \oc_r A \otimes \oc_r B \to \oc_r(A \otimes B)$ and $\mu^I_r: I \to \oc_r I$, and such that the following equalities hold for all ${r,s,t \in R}$: 
\begin{enumerate}[{\em (i)}]
\item $(\oc, \mu^\otimes, \mu^I)$ is an $R$-graded symmetric monoidal functor:
\begin{align} \label{smfeq}
    (\mu^\otimes_r \otimes 1); \mu^\otimes_r = (1 \otimes \mu^\otimes_r); \mu^\otimes_r &&  (\mu^I_r \otimes 1); \mu^\otimes_r = 1= (1 \otimes \mu^I_r); \mu^\otimes_r && \sigma; \mu^\otimes_r  = \mu^\otimes_r ; \oc_r(\sigma)
\end{align}
\item $\mathsf{p}$ and $\mathsf{d}$ are $R$-graded monoidal transformations: 
\begin{align}\label{demeq}
   \mu^\otimes_{r s}; \mathsf{p}_{r,s} = (\mathsf{p}_{r,s} \otimes \mathsf{p}_{r,s}); \mu^\otimes_r; \oc_r(\mu^\otimes_s) && \mu^I_{r s}; \mathsf{p}_{r,s} =  \mu^I_r; \oc_r(\mu^I_s) && \mu^\otimes_1; \mathsf{d} = \mathsf{d} \otimes \mathsf{d} && \mu^I_1; \mathsf{d} = 1 
\end{align}
\item  $\mathsf{c}$ and $\mathsf{d}$ are $R$-graded monoidal transformations: 
\begin{equation}\label{DEmeq}\begin{gathered}
  \mu^\otimes_{r+s}; \mathsf{c}_{r,s} = (\mathsf{c}_{r,s} \otimes \mathsf{c}_{r,s});(1 \otimes \sigma \otimes 1); (\mu^\otimes_r \otimes \mu^\otimes_s)\\
  \mu^I_{r+s}; \mathsf{c}_{r,s} =  (\mu^I_r \otimes \mu^I_s) \quad \quad \quad \mu^\otimes_0; \mathsf{w} = \mathsf{w} \otimes \mathsf{w} \quad \quad \quad \mu^I_0; \mathsf{w} = 1 
\end{gathered}\end{equation}
\item $\mathsf{c}$ and $\mathsf{d}$ are $R$-graded $\oc$-coalgebra morphisms:  
\begin{align} \label{DE!eq}
    \mathsf{p}_{r,s+t}; \oc_r(\mathsf{c}_{s,t}) = \mathsf{c}_{r  s,r t}; (\mathsf{p}_{r,s} \otimes \mathsf{p}_{r,t}); \mu^\otimes_r && \mathsf{p}_{r,0}; \oc_r(e) = \mathsf{w}; \mu^I_r
\end{align}
  \end{enumerate}
\end{definition}

Every non-graded (monoidal) coalgebra modality is trivially graded over the zero semiring:

\begin{example} Let $\mathsf{0}$ be the zero semiring. Then a $\mathsf{0}$-graded (monoidal) coalgebra modality is precisely a non-graded (monoidal) coalgebra modality. 
\end{example}

Here is now a running example we will use throughout the paper of a graded coalgebra modality on the category of sets and relations $\mathsf{REL}$. Recall that $\mathsf{REL}$ is the category whose objects are sets $X$ and whose maps $S: X \to Y$ are subsets $S \subseteq X \times Y$, and also that $\mathsf{REL}$ is a symmetric monoidal category where the monoidal product is given by the Cartesian product $X \otimes Y = X \times Y$ (but this not the categorical product) and the monoidal unit is a chosen singleton $I = \lbrace \ast \rbrace$.  

\begin{example}\label{ex:REL!} Let $\mathbb{N}$ be the semiring of natural numbers. Finite multisets induce an $\mathbb{N}$-graded monoidal coalgebra modality on $\mathsf{REL}$. For a function $f: X \to \mathbb{N}$, define its support as $\mathsf{supp}(f) = \lbrace x \in X ~\vert~ f(x) \neq 0 \rbrace$. If the cardinality of the support of $f$ is finite, $\vert \mathsf{supp}(f) \vert < \infty$, define its degree as $\mathsf{deg}(f) := \sum_{x \in \mathsf{supp}(f)} f(x)$. For every set $X$ and every $n \in \mathbb{N}$, define the set $\oc_n X$ as follows: 
\[\oc_n X = \lbrace f: X \to \mathbb{N} ~\vert~ \vert \mathsf{supp}(f) \vert < \infty \text{ and } \mathsf{deg}(f) = n \rbrace\] 
To define the rest of the structure, we need to first define a few special finite multisets. Define $0: X \to \mathbb{N}$ which maps everything to zero, $0(x) =0$, and note that $\mathsf{deg}(0) = 0$, so $0 \in \oc_0 X$. For each $x \in X$, define $\overline{\mathsf{d}}_x: X \to R$ as $\overline{\mathsf{d}}_x(y) = 0$ if $x \neq y$ and $\overline{\mathsf{d}}_x(x) = 1$, and note that $\mathsf{deg}(\overline{\mathsf{d}}_x) =1$, so $\overline{\mathsf{d}}_x \in \oc_1 X$. For parallel functions $f,g: X \to \mathbb{N}$, define their sum $f+g: X \to \mathbb{N}$ pointwise $(f+g)(x) = f(x) +g(x)$, and note that if $f$ and $g$ have finite support then $\mathsf{deg}(f+g) = \mathsf{deg}(f) + \mathsf{deg}(g)$. We now define the following relations: 
\begin{gather*}
\mathsf{d} = \lbrace (\overline{\mathsf{d}}_x, x) ~\vert~ x \in X \rbrace \subseteq \oc_1 X \times X \quad \quad \quad  
\mathsf{p}_{n,m} = \left\lbrace (f, F) ~\vert~ f, F \text{ s.t. } f = \sum_{f_i \in \mathsf{supp}(F)} f_i \right \rbrace \subseteq \oc_{nm} X \times \oc_n \oc_m X \\
\mathsf{c}_{n,m} =\lbrace (f_1, (f_2,f_3) ) ~\vert~ f_i \text{ s.t. } f_1 = f_2 + f_3 \rbrace \subseteq \oc_{n+m} X \times (\oc_n X \times \oc_m X) \quad \quad \mathsf{w} = \lbrace (0, \ast) \rbrace \subseteq \oc_0 X \times \lbrace \ast \rbrace \\
\mu^\otimes_n = \lbrace ((f_1,f_2), f_3) ~\vert~ f_i \text{ s.t. } f_1 = \sum_{y \in Y} f_3(-,y) \text{ and } f_2 = \sum_{x \in X} f_3(x,-) \rbrace \subseteq (\oc_n X \times \oc_n Y) \times \oc_n(X \times Y)  \\
\mu^I_n = \lbrace (\ast, n) ~\vert~ n \in \mathbb{N} \rbrace \subseteq \lbrace \ast \rbrace \times \oc_n \lbrace \ast \rbrace 
\end{gather*}
Then $(\oc, \mathsf{p}, \mathsf{d}, \mathsf{c}, \mathsf{w}, \mu^\otimes, \mu^I)$ is an $\mathbb{N}$-graded monoidal coalgebra modality on $\mathsf{REL}$ (which follows from the results in Section \ref{sec:sympowers}). 
\end{example}

Other examples of graded monoidal coalgebra modalities can be found in \cite{katsumata2018double}. In Section \ref{sec:sympowers}, we will review how symmetric powers always provide an $\mathbb{N}$-graded monoidal coalgebra modality. 

\section{Graded Differential Categories}\label{sec:gdiffcat}

In this section, we introduce the main novel concept of this paper: graded differential categories, which are graded versions of Blute, Cockett, and Seely's differential categories \cite{blute2006differential}. In particular, this means we will provide a graded version of the \emph{deriving transformation} approach to differential categories. The associated sequent calculus is an extension of $\mathsf{GLL}$ in which there are appropriate notions of zero proofs and sums of proofs, and with an extra differentiation rule ($\partial$): 
\begin{align}
    \begin{prooftree}
\hypo{\Gamma \vdash \oc_{r}A}
\hypo{\Delta \vdash A}
\infer2[$\partial_r$]{\Gamma, \Delta \vdash \oc_{r+1}A}
\end{prooftree}
\end{align}
satisfying the equivalent cut elimination steps for the equations provided below. 

Categorically speaking, sums and zeroes are provided by an additive enrichment of the underlying category. So an \textbf{additive symmetric monoidal category}  \cite[Def 3]{Blute2019} is a symmetric monoidal category $\mathcal{L}$ which is enriched over commutative monoids, that is, each homset $\mathcal{L}(A,B)$ is a commutative monoid with a binary operation $+$, so we can sum parallel maps $f+g$, and zero element $0: A \to B$, such that composition and the monoidal product preserve this additive structure. 

\begin{definition}\label{def:diffcat} An \textbf{$R$-graded differential modality} on an additive symmetric monoidal category $\mathcal{L}$ is a tuple $(\oc, \mathsf{p}, \mathsf{d}, \mathsf{c}, \mathsf{w}, \partial)$ consisting of an $R$-graded coalgebra modality $(\oc, \mathsf{p}, \mathsf{d}, \mathsf{c}, \mathsf{w})$ and a \textbf{deriving transformation}, that is, a family of natural transformations $\partial$ where for every $r \in R$ we have a natural transformation $\partial_r: \oc_r A \otimes A \to \oc_{r+1} A$ and such that the following equalities hold for all $r,s \in R$: 
\begin{enumerate}[{\em (i)}]
\item \textbf{Linear Rule:} $\partial_0; \mathsf{d} = \mathsf{w} \otimes 1$
\item \textbf{Product Rule:} $\partial_{r+s+1}; \mathsf{c}_{r+1, s+1} = (\mathsf{c}_{r+1,s} \otimes 1);(1 \otimes \partial_s) +  (\mathsf{c}_{r,s+1} \otimes 1); (1 \otimes \sigma); (\partial_r \otimes 1)$
\item \textbf{Chain Rule:} $\partial_{rs+r+s};\mathsf{p}_{r+1,s+1} = (\mathsf{c}_{rs+r,s} \otimes 1);(\mathsf{p}_{r,s+1} \otimes \partial_s); \partial_r$
\item \textbf{Symmetry Rule:} $(1 \otimes \sigma);(\partial_r \otimes 1);\partial_r = (\partial_r \otimes 1);\partial_r$
  \end{enumerate}
  An \textbf{$R$-graded differential category} is an additive symmetric monoidal category with a chosen $R$-graded differential modality. 
\end{definition}

The intuition in an $R$-graded differential category is that a non-linear map $\oc_{r+1} A \to B$ should now be interpreted as a \emph{differentiable} map of degree $r+1$. In some models, this means that after differentiating $r+1$ times we get $0$, but this may not always be the case. The derivative of $f: \oc_{r+1} A \to B$ is defined by pre-composing by the deriving transformation, $\mathsf{D}[f] := \partial_r; f : \oc_{r} A \otimes A \to \oc_{r+1} A \to B$, which is interpreted as a non-linear map of degree $r$ in its first argument $A$ and linear in its second argument $A$. The linear rule says that the derivative of a linear map, viewed as a differentiable map of degree $1$, is essentially just itself. The product rule and the chain rule are analogues of their namesakes from classical differential calculus, which tell us how to differentiate products and compositions of non-linear maps respectively. The symmetry rule gives us the symmetry of partial derivatives. To better understand these rules, consider the case $R = \mathbb{N}$ and view $\oc_n A \to B$ as just a homogeneous monomial function of degree $n$, say $x^n$. We then invite the readers to check these formulas for themselves with polynomials in mind. In Section \ref{sec:sympowers}, we will make the intuition of differentiating polynomials precise. 

Readers familiar with differential categories will note that a graded version of the so-called constant rule is missing, that is, $\partial; \mathsf{w} = 0$. In \cite[Lem 3]{Blute2019}, it was shown that the non-graded constant rule follows from the naturality of $\partial$ and $\mathsf{w}$. The same is true in the graded version. Indeed, first note that for $\partial_r; \mathsf{w}$ to be well-typed we need that $r+1=0$, which means that $r=-1$ and that $R$ is, in fact, a ring. Then if $-1$ exists in $R$, we indeed have the constant rule that is $\partial_{-1}; \mathsf{w} = 0$, which entails that the derivative of points (which we can also understand as constant maps), viewed as a differential map of degree $0$, is zero. 

\begin{lemma} Let $R$ be a ring. Then in a $R$-graded differential category, the following equality holds: 
\begin{enumerate}[{\em (i)}]
\item \textbf{Constant Rule:} $\partial_{-1}; \mathsf{w} = 0$
  \end{enumerate}
\end{lemma}
\begin{proof} By naturality of $\mathsf{w}$ and $\partial$, i.e $\mathsf{w} = \oc_0(f); \mathsf{w}$ and $\partial_r;\oc_r(f) = (\oc_r(f) \otimes f);\partial_r$, and that the monoidal product preserves zero, i.e.\! $f \otimes 0 =0$, we easily compute that: 
\[ \partial_{-1}; \mathsf{w} = \partial_{-1};\oc_0 (0); \mathsf{w} = (\oc_{-1}(0) \otimes 0); \partial_{-1}; \mathsf{w} = 0 \]
So we conclude that $\partial_{-1}; \mathsf{w} = 0$. 
\end{proof}

We now define deriving transformations for graded monoidal coalgebra modalities by also requiring that a graded version of Fiore's monoidal rule from \cite{fiore2007differential} holds. Again, readers familiar with differential categories will note that in the non-graded case, the monoidal rule can be proved from the other deriving transformation axioms, and is, in fact, equivalent to the product rule \cite[Prop 7]{Blute2019}. However, we claim that in the graded case, this is no longer true. This is because the proof in the non-graded uses a construction that is no longer well-typed in the graded setting. We will discuss more about this in the next section.  

\begin{definition}\label{def:diffmoncat} An \textbf{$R$-graded monoidal differential modality} on an additive symmetric monoidal category $\mathcal{L}$ is a tuple $(\oc, \mathsf{p}, \mathsf{d}, \mathsf{c}, \mathsf{w}, \mu^\otimes, \mu^I, \partial)$ consisting of an $R$-graded monoidal coalgebra modality and a \textbf{monoidal deriving transformation}, which is a deriving transformation $\partial$ which also satisfies the following equality for all $r \in R$: 
\begin{enumerate}[{\em (i)}]
\setcounter{enumi}{4}
\item \textbf{Monoidal Rule:} $(1 \otimes \partial_r); \mu^\otimes_{r+1} = (\mathsf{c}_{r+1} \otimes 1 \otimes 1);(1 \otimes \mathsf{d} \otimes 1 \otimes 1);(1 \otimes \sigma \otimes 1);(\mu^\otimes_r \otimes 1);\partial_r$
  \end{enumerate}
\end{definition}

We conclude this section with our running example. 

\begin{example} A $\mathsf{0}$-graded (monoidal) differential modality is precisely the same as a non-graded (monoidal) differential modality. Therefore, a $\mathsf{0}$-graded differential category is precisely a non-graded differential category. 
\end{example}

\begin{example}\label{ex:RELd} $\mathsf{REL}$ is an additive symmetric monoidal category where the sum of relations is given by union, $+ = \sqcup$, and the zero map is given by the empty relation, $0 = \emptyset$. Define the deriving transformation for the $\mathbb{N}$-graded monoidal coalgebra modality from Ex \ref{ex:REL!} as the following relation: 
\begin{align*}
    \partial_n = \lbrace ((f, x), f + \overline{\mathsf{d}}_x ) ~\vert~ x\in X, f \in \oc_n X \rbrace \subseteq (\oc_n X \times X) \times \oc_{n+1} X 
\end{align*}
Then $(\oc, \mathsf{p}, \mathsf{d}, \mathsf{c}, \mathsf{w}, \mu^\otimes, \mu^I, \partial)$ is an $\mathbb{N}$-graded monoidal differential modality and so $\mathsf{REL}$ is an $\mathbb{N}$-graded differential category (which again follows from the results in Section \ref{sec:sympowers}).  
\end{example}

\section{Graded Additive Bialgebra Modalities and Seely Isomorphisms}\label{sec:gbialg}

In order to properly define a graded version of a codereliction, we first need to discuss graded versions of cocontraction and coweakening. In the sequent calculus, these are two extra rules for the graded exponential modality, the \textbf{cocontraction} rule ($\overline{\mathsf{c}}$) and the \textbf{coweakening} rule ($\overline{\mathsf{w}}$): 
\begin{align}
\begin{prooftree}
\hypo{\Gamma \vdash \oc_{r}A}
\hypo{\Delta \vdash \oc_{s}A}
\infer2[$\overline{\mathsf{c}}_{r,s}$]{\Gamma, \Delta \vdash \oc_{r+s}A}
\end{prooftree}
&&
\begin{prooftree}
\infer0[$\overline{\mathsf{w}}$]{\vdash \oc_{0}A}
\end{prooftree}
\end{align}
In categorical terms, the coweakening rule is given by a natural transformation of type $\overline{\mathsf{w}}: I \to \oc_0 A$, while the cocontraction rule is given by an $R$-indexed family of natural transformations $\overline{\mathsf{c}}_{r,s}: \oc_r A \otimes \oc_s A \to \oc_{r+s} A$. Note that $\overline{\mathsf{c}}$ and $\overline{\mathsf{w}}$ are of dual types of $\mathsf{c}$ and $\mathsf{w}$. Together, these four natural transformations subject to equations that are presented below make $\oc_r A$ into an $R$-graded bimonoid (the generalization of a bialgebra in a symmetric monoidal category). 

In classical algebra, the notion of graded bialgebra is usually defined simply as being graded over a monoid $R$, typically $R = \mathbb{N}$, and as a bialgebra $A$ which can be decomposed as a direct sum $A = \bigoplus_{r\in R}A_{r}$ such that the multiplication and comultiplication are compatible in appropriate ways with this grading \cite[Chap XI]{sweedler1969hopf}. In our context, $R$ is a semiring and we replace the underlying object $A$ by the family $(A_{r})_{r \in R}$ of its components without assuming biproducts. The compatibility between the comultiplication and the multiplication now uses a sum indexed by elements of the semiring that verify some equations. This sum needs to be finite, and in order for this to hold, we must require that there are only a finite number of elements which sum to $r$. More precisely:

\begin{definition} A semiring $R$ is \textbf{finite additive split} (f.a.s) if the following conditions hold: 
\begin{enumerate}[{\em (i)}]
\item For every $r \in R$, the set $+^{-1}(r) = \lbrace (s,t) ~\vert~ \forall s,t \in R \text{ s.t. } s+t = r \rbrace$ is finite. In other words, for every $r \in R$, there is only a finite number of elements $s,t\in R$ such that $s+t=r$.
\item We have that $+^{-1}(0) = \lbrace (0,0)  \rbrace$ and $+^{-1}(1) = \lbrace (1,0), (0,1) \rbrace$. In other words, if $s+t=0$ then $s=t=0$, and if $s+t=1$, then either $s=0$ and $t=1$, or $s=1$ and $t=0$. 
\item $R$ is cancellative, that is, if $r+s = r+t$ then $s=t$. 
\end{enumerate}
\end{definition}

The axioms for a f.a.s.\! semiring are similar to those of a multiplicity semiring \cite[Sec 2.1]{carraro2010exponentials}, though the axioms here are slightly stronger. Examples of f.a.s.\! semirings include $\mathsf{0}$, $\mathbb{N}$, $\mathbb{N} \times \mathbb{N}$, $\mathbb{N}[x]$, etc.  

We now provide the definition of a graded additive bialgebra modality, which is the graded version of \cite[Def 5]{Blute2019}. Notice that the following definition does not assume that the underlying graded coalgebra modality is monoidal (this will come later).

\begin{definition}\label{def:addbialg} Let $R$ be a f.a.s.\! semiring. An \textbf{$R$-graded additive bialgebra modality} on an additive symmetric monoidal category $\mathcal{L}$ is a tuple $(\oc,\mathsf{p},\mathsf{d},\mathsf{c}, \mathsf{w}, \overline{\mathsf{c}}, \overline{\mathsf{w}})$ consisting of an $R$-graded coalgebra modality $(\oc,\mathsf{p},\mathsf{d},\mathsf{c}, \mathsf{w})$, a family of natural transformation $\overline{\mathsf{c}}$ where for every $r,s \in R$ we have a natural transformation of type $\overline{\mathsf{c}}_{r,s}: \oc_{r}A \otimes \oc_{s}A \to \oc_{r+s}A$, and a natural transformation $\overline{\mathsf{w}} : I \rightarrow \oc_{0}A$ such that the following equalities hold: 
\begin{enumerate}[{\em (i)}]
\item $(\oc_{\_}A, \mathsf{c}, \mathsf{w})$ is an $R$-graded cocommutative comonoid: 
\begin{align} \label{monoideq}
 (\overline{\mathsf{c}}_{r,s} \otimes 1);\overline{\mathsf{c}}_{r+s,t} = (1 \otimes \overline{\mathsf{c}}_{s,t});\overline{\mathsf{c}}_{r,s+t} && (\overline{\mathsf{w}} \otimes 1);\overline{\mathsf{c}}_{0,r}=1=(1 \otimes \overline{\mathsf{w}}); \overline{\mathsf{c}}_{r,0} && \sigma;\overline{\mathsf{c}}_{s,r} = \overline{\mathsf{c}}_{r,s}
\end{align}
\item $(\oc_{\_}A, \overline{\mathsf{c}}, \overline{\mathsf{w}}, \mathsf{c}, \mathsf{w})$ is an $R$-graded bimonoid in the sense that: 
\begin{equation}\label{bialgeq}\begin{gathered} 
\overline{\mathsf{c}}_{r,s}; \mathsf{c}_{t,u} = \sum\limits_{\substack{a,b,c,d \in R \\ a+b = r \quad c+d = s \\ a+c = t \quad b+d = u}} (\mathsf{c}_{a,b} \otimes \mathsf{c}_{c,d}); (1 \otimes \sigma \otimes 1); \overline{\mathsf{c}}_{a,c} \otimes \overline{\mathsf{c}}_{b,d} \quad \quad \quad \text{ where } r+s = t+u \\
   \overline{\mathsf{w}};\mathsf{w} = 1 \quad \quad \quad \overline{\mathsf{w}}; \mathsf{c}_{0,0} = \overline{\mathsf{w}} \otimes \overline{\mathsf{w}} \quad \quad \quad \overline{\mathsf{c}}_{0,0}; \mathsf{w}= \mathsf{w} \otimes \mathsf{w}
\end{gathered}\end{equation}
\item $\mathsf{d}$ is compatible with $\overline{\mathsf{c}}$ in the sense that: 
\begin{align}\label{dceq}
    \overline{\mathsf{c}}_{r,s}; \mathsf{d} = \delta_{r,0} \cdot (\mathsf{w} \otimes \mathsf{d}) + \delta_{0,s} \cdot (\mathsf{d} \otimes \mathsf{w}) && \text{ where } r+s = 1
\end{align}
where $\delta_{x,y} \in R$ is Kronecker delta, that is, $\delta_{x,y}=1$ if $x=y$ and $\delta_{x,y}=0$ if $x\neq y$. 
\item For every map $f: A \to B$, the following equalities hold: 
\begin{align}\label{!feq}
    \oc_r(f+g)= \sum\limits_{\substack{s,t \in R \\ s+t=r}} \mathsf{c}_{s,t}; (\oc_s (f) \otimes \oc_t(g)); \overline{\mathsf{c}}_{s,t} && \oc_r(0) = \begin{cases}  \mathsf{w};\overline{\mathsf{w}} & \text{ if } r=0 \\
0 & \text{ o.w. }
\end{cases}
\end{align}
\end{enumerate}
\end{definition}

Let us explain these axioms a bit more. The equations of (\ref{bialgeq}) are generalizations of those for $\mathbb{N}$-graded bialgebras \cite[Lem 4.1]{ardizzoni2012associated} but for f.a.s.\! semirings -- which are well-defined since the sum needs to be finite. Of course, if we are working in an additive symmetric monoidal category where infinite sums are well-defined, then the assumption that $R$ is an f.a.s.\! could, in theory, be dropped. When $R=\mathsf{0}$, there is only one term in the sum, so we re-obtain the non-graded bimonoid axiom. The equations of (\ref{!feq}) say that $\oc$ relates the additive enrichment to the graded bimonoid convolution. Again when $R=\mathsf{0}$, there is only one term in the sum, so we re-obtain the bimonoid convolution axiom in the non-graded case. At first glance, (\ref{dceq}) might seem unproperly typed since $\mathsf{w} \otimes \mathsf{d}: \oc_0 A \otimes \oc_1 A \to A$ and $\mathsf{d} \otimes \mathsf{w}: \oc_1 A \otimes \oc_0 A \to A$. However, when $R \neq \mathsf{0}$, $\delta_{r,0} \cdot (\mathsf{w} \otimes \mathsf{d}) + \delta_{0,s} \cdot (\mathsf{d} \otimes \mathsf{w})$ is either $\mathsf{w} \otimes \mathsf{d}$ when $r=0$ and $s=1$, or $\mathsf{d} \otimes \mathsf{w}$ when $r=1$ or $s=0$, or $0$ when $r$ and $s$ are neither $0$ or $1$. When $R = \mathsf{0}$, then $0=1$ and therefore we obtain the sum $\mathsf{w} \otimes \mathsf{d} + \mathsf{d} \otimes \mathsf{w}$, which is the axiom for the non-graded case. 

\begin{example} A $\mathsf{0}$-graded additive bialgebra modality is precisely a non-graded additive bialgebra modality. 
\end{example}

\begin{example}\label{ex:RELbialg} Since $\mathbb{N}$ is a f.a.s. semiring, define the following relations: 
\begin{gather*}
\overline{\mathsf{c}}_{n,m} =\lbrace ((f_1,f_2) ,f_3 ) ~\vert~ f_i \text{ s.t. } f_1 + f_2 = f_3 \rbrace \subseteq (\oc_n X \times \oc_m X) \times \oc_{n+m} X \quad \quad \overline{\mathsf{w}} = \lbrace (\ast, 0) \rbrace \subseteq  \lbrace \ast \rbrace \times \oc_0 X 
\end{gather*}
Then $(\oc,\mathsf{p},\mathsf{d},\mathsf{c}, \mathsf{w}, \overline{\mathsf{c}}, \overline{\mathsf{w}})$ is an $\mathbb{N}$-graded additive bialgebra modality on $\mathsf{REL}$ (which again follows from the results in Section \ref{sec:sympowers}).  
\end{example}

When working in the additive fragment of $\mathsf{GLL}$, it turns out that we obtain graded versions of the Seely isomorphisms. In a setting with sums of proofs, note that the additive disjunction $\oplus$ and the additive conjunction $\with$ are actually the same, so $\oplus = \with$ \cite{fiore2007differential}. Thus writing $\oplus$ for both, recall the following rules: 
\begin{gather*}
\begin{prooftree}
\hypo{\Gamma, A \vdash C}
\hypo{\Gamma, B \vdash C}
\infer2[$\oplus_{l}$]{\Gamma, A \oplus B \vdash C}
\end{prooftree}
\quad \quad 
\begin{prooftree}
\hypo{\Gamma \vdash A}
\infer1[$\oplus_{r}^{1}$]{\Gamma \vdash A \oplus B}
\end{prooftree}
\quad \quad 
\begin{prooftree}
\hypo{\Gamma \vdash B}
\infer1[$\oplus_{r}^{2}$]{\Gamma \vdash A \oplus B}
\end{prooftree}
\\
\begin{prooftree}
\hypo{\Gamma \vdash A}
\hypo{\Gamma \vdash B}
\infer2[$\with_{r}$]{\Gamma \vdash A \oplus B}
\end{prooftree}
\quad \quad 
\begin{prooftree}
\hypo{\Gamma, A \vdash C}
\infer1[$\with_{l}^{1}$]{\Gamma, A \oplus B \vdash C}
\end{prooftree}
\quad \quad  
\begin{prooftree}
\hypo{\Gamma, B \vdash C}
\infer1[$\with_{l}^{2}$]{\Gamma, A \oplus B \vdash C}
\end{prooftree}
\end{gather*}
Then using these rules, in addition to (co)contraction, dereliction and promotion we derive:  
\begin{align*}
\begin{prooftree}
\infer0[$\mathsf{ax}$]{A \vdash A}
\infer1[$\with_{l}^{1}$]{A \oplus B \vdash A}
\infer1[$\mathsf{d}$]{\oc_{1}(A \oplus B) \vdash A}
\infer1[$\mathsf{prom}_{s,1}$]{\oc_{s}(A \oplus B) \vdash \oc_{s}A}
\infer0[$\mathsf{ax}$]{B \vdash B}
\infer1[$\with_{l}^{2}$]{A \oplus B \vdash B}
\infer1[$\mathsf{d}$]{\oc_{1}(A \oplus B) \vdash B}
\infer1[$\mathsf{prom}_{t,1}$]{\oc_{t}(A \oplus B) \vdash \oc_{t}B}
\infer2[$\otimes_{r}$]{\oc_{s}(A \oplus B), \oc_{t}(A \oplus B) \vdash \oc_{s}A \otimes \oc_{t}B}
\infer1[$\mathsf{c}_{s,t}$]{\oc_{r}(A \oplus B) \vdash \oc_{s}A \otimes \oc_{t}B}
\infer1[$\with_{r}$]{\oc_{r}(A \oplus B) \vdash \bigoplus\limits_{\substack{s,t \in R \\ s+t =r}}\oc_{s}A \otimes \oc_{t}B}
\end{prooftree}
&& 
\begin{prooftree}
\infer0[$\mathsf{ax}$]{A \vdash A}
\infer1[$\oplus_{r}^{1}$]{A \vdash A \oplus B}
\infer1[$\mathsf{d}$]{\oc_{1}A \vdash A \oplus B}
\infer1[$\mathsf{prom}_{s,1}$]{\oc_{s}A \vdash \oc_{s}(A \oplus B)}
\infer0[$\mathsf{ax}$]{B \vdash B}
\infer1[$\oplus_{r}^{2}$]{B \vdash A\oplus B}
\infer1[$\mathsf{d}$]{\oc_{1}B \vdash A \oplus B}
\infer1[$\mathsf{prom}_{t,1}$]{\oc_{t}B \vdash \oc_{t}(A \oplus B)}
\infer2[$\overline{\mathsf{c}}_{s,t}$]{\oc_{s}A, \oc_{t}B \vdash \oc_{r}(A \oplus B)}
\infer1[$\otimes_{l}$]{\oc_{s}A \otimes \oc_{t}B \vdash \oc_{r}(A \oplus B)}
\infer1[$\oplus_{l}$]{\bigoplus\limits_{\substack{s,t \in R \\ s+t =r}} \oc_{s}A \otimes \oc_{t}B \vdash \oc_{r}(A \oplus B)}
\end{prooftree}
\end{align*}
where the value $r$ is fixed and the last rules are obtained by repeating a finite number of times the corresponding binary rule for every pair of elements $s,t \in R$ such that $s+t=r$ -- which is indeed a finite number of times since $R$ is f.a.s.\! These two derivations are inverses of each other. 

Categorically speaking, recall that the additive fragment, in this case, is captured by finite biproducts \cite[Sec 2]{fiore2007differential}. So for a category $\vert \mathcal{L} \vert$ with finite biproducts, we denote the biproduct as $\oplus$, the projections by ${\pi_j: A_0 \oplus \hdots \oplus A_n \to A_j}$, the injections by $\iota_j: A_j \to A_0 \oplus \hdots \oplus A_n$, and the zero object by $\mathsf{0}$. So we will now prove that the above derivations correspond to isomorphisms and obtain: 
\begin{align}
    \oc_r(A \oplus B) \cong \bigoplus\limits_{\substack{s,t \in R \\ s+t =r}}  \oc_s A \otimes \oc_t B && \oc_0 \mathsf{0} \cong I
\end{align}
We are making a slight abuse of notation here and allow for indexing over the set $+^{-1}(r)$, which for an f.a.s.\! semiring is finite. We denote by $\iota_{s,t}$ and $\pi_{s,t}$ the injection and projection into the $(s,t) \in +^{-1}(r)$ term of the biproduct. Also recall that for biproducts, the pairing and copairing operations (so $\with_r$ and $\oplus_l$) can be defined using sums. 

\begin{theorem}\label{thm:Seely} Let $R$ be a f.a.s.\! semiring and let $(\oc,\mathsf{p},\mathsf{d},\mathsf{c}, \mathsf{w}, \overline{\mathsf{c}}, \overline{\mathsf{w}})$ be an $R$-graded additive bialgebra modality on an additive symmetric monoidal category $\mathcal{L}$ with finite biproducts. Define the natural transformations $\chi^\otimes_r: \oc_r(A \oplus B) \to \bigoplus\limits_{\substack{s,t \in R \\ s+t =r}}  \oc_s A \otimes \oc_t B$ and $\chi^I: \oc_0 \mathsf{0} \to I$ respectively as follows: 
\begin{gather*}
\chi^\otimes_r :=\begin{array}[c]{c}
\sum\limits_{\substack{s,t \in R \\ s+t =r}}
   \end{array} \left( \begin{array}[c]{c} \xymatrixcolsep{5pc}\xymatrix{ \oc_r(A \oplus B) \ar[r]^-{\mathsf{c}_{s,t}} & \oc_s (A \oplus B) \otimes \oc_t (A \oplus B) \ar[d]^-{\oc_s(\pi_0) \otimes \oc_t(\pi_1)} & \\ &  \oc_s A \otimes \oc_t B \ar[r]^-{\iota_{s,t}} &\bigoplus\limits_{\substack{s,t \in R \\ s+t =r}}  \oc_s A \otimes \oc_t B
  } \end{array}\right) \\
  \chi^I := \mathsf{w}
\end{gather*} 
Then $\chi^\otimes_r$ and $\chi^I$ are natural isomorphisms with inverses defined as follows: 
\begin{gather*}
  {\chi^\otimes_{r}}^{-1} :=\begin{array}[c]{c}
\sum\limits_{\substack{s,t \in R \\ s+t =r}}
   \end{array} \left(\begin{array}[c]{c} \xymatrixcolsep{5pc}\xymatrix{ \bigoplus\limits_{\substack{s,t \in R \\ s+t =r}}  \oc_s A \otimes \oc_t B \ar[r]^-{\pi_{s,t}} & \oc_s A \otimes \oc_t B \ar[d]^-{\oc_s(\iota_0) \otimes \oc_t(\iota_1)} & \\ & \oc_{s}(A \oplus B) \otimes \oc_{t}(A \oplus B)  \ar[r]^-{\overline{\mathsf{c}}_{s,t}} & \oc_{r} (A \oplus B)
  }  \end{array}\right) \\
  {\chi^I}^{-1} := \overline{\mathsf{w}} 
\end{gather*}  
\end{theorem}
\begin{proof} For starters, by one of the axioms in  (\ref{bialgeq}), we already have that: 
\[{\chi^I}^{-1};\chi^I = \overline{\mathsf{w}} ; \mathsf{w} \underset{(\ref{bialgeq})}{=} 1 \]
For the zero object, $1 =0$, therefore by (\ref{!feq}), we also have that: 
\[\chi^I; {\chi^I}^{-1}= \mathsf{w};\overline{\mathsf{w}} \underset{(\ref{!feq})}{=} \oc_0(0) = \oc_0(1) = 1\]
To show that $\chi^\otimes_r$ and ${\chi^\otimes_{r}}^{-1}$ inverses, we need to use the biproduct identities that $\iota_{j};\pi_{i}=\delta_{j,i}$ and $\sum_j \pi_j ; \iota_j = 1$. Now for $\chi^\otimes_r ; {\chi^\otimes_{r}}^{-1}=1$, we use the first biproduct identity to first greatly reduce the sum, and then use (\ref{!feq}) to afterwards use the second biproduct identity: 
\begin{gather*}
  \chi^\otimes_r ; {\chi^\otimes_{r}}^{-1} = \left( \sum\limits_{s+t =r} \mathsf{c}_{s,t}; \left( \oc_s(\pi_0) \otimes \oc_t(\pi_1) \right); \iota_{s,t}  \right); \left(\sum\limits_{s^\prime+t^\prime =r} \pi_{s^\prime,t^\prime}; \left( \oc_{s^\prime}(\iota_0) \otimes \oc_{t^\prime}(\iota_1) \right);  \overline{\mathsf{c}}_{s^\prime,t^\prime} \right) \\
  = \sum\limits_{s+t =r} \mathsf{c}_{s,t}; \left( \oc_s(\pi_0) \otimes \oc_t(\pi_1) \right); \left( \oc_{s^\prime}(\iota_0) \otimes \oc_{t^\prime}(\iota_1) \right); \overline{\mathsf{c}}_{s,t} \underset{(\ref{!feq})}{=} \oc_r\left( \pi_0;\iota_0 + \pi_1;\iota_1 \right) = \oc_r(1) = 1
\end{gather*}
For ${\chi^\otimes_{r+s}}^{-1};\chi^\otimes_r =1$, we first use the graded bialgebra axiom (\ref{bialgeq}), then using naturality and the first biproduct identity, we obtain $\oc_a(0)$, we can then greatly simplify the sum using (\ref{!feq}) in order to use the second biproduct identity: 
\begin{gather*}
      {\chi^\otimes_{r}}^{-1}; \chi^\otimes_r = \left(\sum\limits_{s^\prime+t^\prime =r} \pi_{s^\prime,t^\prime}; \left( \oc_{s^\prime}(\iota_0) \otimes \oc_{t^\prime}(\iota_1) \right);  \overline{\mathsf{c}}_{s^\prime,t^\prime} \right); \left( \sum\limits_{s+t =r} \mathsf{c}_{s,t}; \left( \oc_s(\pi_0) \otimes \oc_t(\pi_1) \right); \iota_{s,t}  \right) \\
      = \sum\limits_{\substack{s+t =r \\ s^\prime + t^\prime = r}} \pi_{s^\prime,t^\prime}; \left( \oc_{s^\prime}(\iota_0) \otimes \oc_{t^\prime}(\iota_1) \right);  \overline{\mathsf{c}}_{s^\prime,t^\prime}; \mathsf{c}_{s,t}; \left( \oc_s(\pi_0) \otimes \oc_t(\pi_1) \right); \iota_{s,t} \\
\underset{(\ref{bialgeq})}{=} \sum\limits_{\substack{s+t =r \\ s^\prime + t^\prime = r}}  \sum\limits_{\substack{a+b = s \quad c+d = t \\ a+c = s^\prime \quad b+d = t^\prime}} \pi_{s^\prime,t^\prime}; \left( \oc_{s^\prime}(\iota_0) \otimes \oc_{t^\prime}(\iota_1) \right);(\mathsf{c}_{a,b} \otimes \mathsf{c}_{c,d}); (1 \otimes \sigma \otimes 1); \overline{\mathsf{c}}_{a,c} \otimes \overline{\mathsf{c}}_{b,d}; \left( \oc_s(\pi_0) \otimes \oc_t(\pi_1) \right); \iota_{s,t} \\
= \sum\limits_{\substack{s+t =r \\ s^\prime + t^\prime = r}}  \sum\limits_{\substack{a+b = s \quad c+d = t \\ a+c = s^\prime \quad b+d = t^\prime}} \pi_{s^\prime,t^\prime};(\mathsf{c}_{a,b} \otimes \mathsf{c}_{c,d}); (\oc_a(1) \otimes \oc_b(0) \otimes \oc_c(0) \otimes \oc_d(1)); (1 \otimes \sigma \otimes 1); \overline{\mathsf{c}}_{a,c} \otimes \overline{\mathsf{c}}_{b,d}; \iota_{s,t} \\
\underset{(\ref{!feq})}{=} \sum\limits_{s+t =r}  \pi_{s, t};(\mathsf{c}_{a,b} \otimes \mathsf{c}_{c,d}); (\oc_s(1) \otimes \oc_0(0) \otimes \oc_0(0) \otimes \oc_t(1)); (1 \otimes \sigma \otimes 1); \overline{\mathsf{c}}_{a,c} \otimes \overline{\mathsf{c}}_{b,d}; \iota_{s,t}  \underset{(\ref{!feq})}{=}  \sum\limits_{s+t =r}  \pi_{s, t}; \iota_{s,t} = 1
\end{gather*}    
Thus the graded Seely isomorphisms hold as desired. 
\end{proof}

The converse of Thm \ref{thm:Seely} is also true. Indeed, note that $\chi^\otimes_r$ and $\chi^I$ can be defined for any graded coalgebra modality, and we could ask that they be natural isomorphisms, giving us a notion of graded storage coalgebra modality -- the graded version of \cite[Def 10]{Blute2019}. Then it follows that we can construct both $\overline{\mathsf{c}}$ and $\overline{\mathsf{w}}$, and obtain a graded additive bialgebra modality. 

\begin{proposition}\label{prop:Seely2} Let $R$ be a f.a.s.\! semiring and let $(\oc,\mathsf{p},\mathsf{d},\mathsf{c}, \mathsf{w})$ be an $R$-graded coalgebra modality on an additive symmetric monoidal category $\mathcal{L}$ with finite biproducts. Suppose that the natural transformations $\chi^\otimes_r$ and $\chi^I$, as defined in Thm \ref{thm:Seely}, are natural isomorphisms. Define: 
    \begin{gather*}
\overline{\mathsf{c}}_{r,s} :=  \xymatrixcolsep{4pc}\xymatrix{ \oc_r A \otimes \oc_s A \ar[r]^-{\iota_{r,s}} & \bigoplus\limits_{\substack{a,b \in R \\ a+b = r+s}}  \oc_a A \otimes \oc_b B \ar[r]^-{{\chi^\otimes_{r+s}}^{-1}} & \oc_{r+s}(A \oplus A) \ar[r]^-{\oc_{r+s}(\nabla)} & \oc_{r+s} A
  } \\
  \overline{\mathsf{w}} := \xymatrixcolsep{4pc}\xymatrix{I \ar[r]^-{{\chi^I}^{-1}} & \oc_0 \mathsf{0} \ar[r]^-{\oc_0(0)} & \oc_0 A 
  }
\end{gather*}
where $\nabla: A \oplus A \to A$ is the codiagonal map of the biproduct. Then $(\oc,\mathsf{p},\mathsf{d},\mathsf{c}, \mathsf{w}, \overline{\mathsf{c}}, \overline{\mathsf{w}})$ is an $R$-graded additive bialgebra modality. 
\end{proposition}
\begin{proof} The proof is essentially the same as the non-graded version \cite[Sec 7]{Blute2019}, which can be done mostly by brute force calculations. The main idea for this proof is that $\oc$ transfers the canonical bimonoid structure on $\oplus$ to a graded bimonoid structure on $\otimes$. 
\end{proof}

It is worth mentioning that the constructions of Thm \ref{thm:Seely} and Prop \ref{prop:Seely2} are inverses of each other. In particular, if we start with a $R$-graded additive bialgebra modality $(\oc,\mathsf{p},\mathsf{d},\mathsf{c}, \mathsf{w}, \overline{\mathsf{c}}, \overline{\mathsf{w}})$, and apply the constructions of Prop \ref{prop:Seely2} we get back precisely $\overline{\mathsf{c}}$ and $\overline{\mathsf{u}}$. The proof is again similar to the non-graded version  \cite[Sec 7]{Blute2019}, where in this case one uses the graded bialgebra axioms (\ref{bialgeq}) and (\ref{!feq}), as well as the biproduct identities. 

As mentioned, the definition of a graded additive bialgebra modality does not require the underlying graded coalgebra modality to be monoidal. The same is true in the non-graded setting. In fact, in the non-graded setting it can be shown that every additive bialgebra modality is also a monoidal coalgebra modality, and vice-versa \cite[Thm 1]{Blute2019}. However the non-graded constructions of going from an additive bialgebra modality to a monoidal coalgebra modality \cite[Prop 2]{Blute2019}, or vice-versa \cite[Prop 1]{Blute2019}, use $\mathsf{d}$ and $\mathsf{w}$. As such, these constructions are no longer properly typed in the graded setting since $\mathsf{d}$ and $\mathsf{w}$ are now only well-typed for $\oc_1$ and $\oc_0$ respectively. Here are now counter-examples to further justify why graded monoidal coalgebra modalities are not the same as graded additive bialgebra modalities: 

\begin{example} Observe that for any semiring $R$, any non-graded monoidal coalgebra modality is an $R$-graded monoidal coalgebra modality by setting $\oc_r = \oc$ for all $r \in R$. Now take $R=\mathbb{N}$ and suppose that we are in an additive symmetric monoidal category with finite biproducts. If $\oc$ was also an $\mathbb{N}$-graded additive bialgebra modality, then by Thm \ref{thm:Seely},  we would have that $\oc(A \oplus B) = \oc_1(A \oplus B) \cong (\oc_1 A \otimes \oc_0 B) \oplus (\oc_1 A \otimes \oc_0 B) = (\oc A \otimes \oc B) \oplus (\oc A \otimes \oc B)$. But this is false since the non-graded version of the Seely isomorphism gives us $\oc(A \oplus B) \cong \oc A \otimes \oc B$ only. So $\oc$ while is a $\mathbb{N}$-graded monoidal coalgebra modality, $\oc$ cannot be an $\mathbb{N}$-graded additive bialgebra modality. 
\end{example}

\begin{example} For any additive symmetric monoidal category with biproducts, we have an $\mathbb{N}$-graded additive bialgebra modality by setting $\oc_0 A = I$ and $\oc_{n+1} A = \mathsf{0}$ for all $n \in \mathbb{N}$, and the rest of the structure is quite trivial. The promotion is $\mathsf{p}_{0,n}= 1_I$ and $\mathsf{p}_{n+1,0}= 0$, while the dereliction is simply $\mathsf{d}=0$. The counit and unit are both $\mathsf{w}=\overline{\mathsf{w}}=1$. The comultiplication and multiplication are $\mathsf{c}_{0,0}=1=\overline{\mathsf{c}}_{0,0}$ and zero otherwise. However, $\oc$ is not monoidal since this would in particular require a map $\mu^I_1: I \to \oc_1 I$ such that $\mu^I_1;\mathsf{d}=1$. But since $\mathsf{d}=0$, this would mean that $0=1$, which is not the case since $I$ is not a zero object (if it was, then the additive symmetric monoidal category is trivial). Thus $\oc$ is an $\mathbb{N}$-graded additive bialgebra modality which is not a $\mathbb{N}$-graded monoidal coalgebra modality. 
\end{example}

We now introduce the monoidal version of graded additive bialgebra modalities, by requiring extra compatibility relations between the graded monoid structure and the graded monoidal structure. These extra axioms are required if we wish to still have a bijective correspondence between deriving transformations and coderelictions in the graded setting (Thm \ref{thm:der-coder}). These axioms are graded versions of \cite[Prop 2]{Blute2019} (which were first observed by Fiore in \cite{fiore2007differential}): 

\begin{definition}\label{def:monaddbialg} Let $R$ be a f.a.s.\! semiring. An \textbf{$R$-graded monoidal additive bialgebra modality} on an additive symmetric monoidal category $\mathcal{L}$ is a tuple $(\oc,\mathsf{p},\mathsf{d},\mathsf{c}, \mathsf{w},\mu^\otimes, \mu^I, \overline{\mathsf{c}}, \overline{\mathsf{w}})$ consisting of an $R$-graded monoidal coalgebra modality $(\oc,\mathsf{p},\mathsf{d},\mathsf{c}, \mathsf{w}, \mu^\otimes, \mu^I)$ and an $R$-graded additive bialgebra modality $(\oc,\mathsf{p},\mathsf{d},\mathsf{c}, \mathsf{w}, \overline{\mathsf{c}}, \overline{\mathsf{w}})$, such that the following equalities hold:  
\begin{enumerate}[{\em (i)}]
\item $\overline{\mathsf{c}}$ and $\overline{\mathsf{w}}$ are $R$-graded $\oc$-coalgebra morphisms: 
\begin{align} \label{cwcoalgeq} 
    \overline{\mathsf{c}}_{rs,rs}; \mathsf{p}_{r,s+t} = (\mathsf{p}_{r,s} \otimes \mathsf{p}_{r,t});\mu^\otimes_r;\oc_r(\overline{\mathsf{c}}_{s,t}) && \overline{\mathsf{w}};\mathsf{p}_{r,0} = \mu^I_r; \oc(\overline{\mathsf{w}})
\end{align}
\item $\mu^\otimes$ is also compatible with $\overline{\mathsf{c}}$ and $\overline{\mathsf{w}}$ in the following sense: 
\begin{align} \label{cwmuexeq}
    (1 \otimes \overline{\mathsf{c}}_{r,s}); \mu^\otimes_{r+s}= (\mathsf{c}_{r,s} \otimes 1 \otimes 1);(1 \otimes \sigma \otimes 1);(\mu^\otimes_r \otimes\mu^\otimes_s);\overline{\mathsf{c}}_{r,s} && (1 \otimes \overline{\mathsf{w}}); \mu^\otimes_0= \mathsf{w};  \overline{\mathsf{w}}
\end{align}
\end{enumerate}
\end{definition}

\begin{example} Every $\mathsf{0}$-graded (i.e.\! non-graded) additive bialgebra modality is always a $\mathsf{0}$-graded monoidal additive bialgebra modality. Since for an additive symmetric monoidal category, non-graded monoidal coalgebra modalities are always non-graded additive bialgebra modalities \cite[Thm 1]{Blute2019}, we may also alternatively state that every $\mathsf{0}$-graded monoidal coalgebra modality is always a $\mathsf{0}$-graded monoidal additive bialgebra modality. 
\end{example}

\begin{example}\label{ex:RELaddbialg} $(\oc,\mathsf{p},\mathsf{d},\mathsf{c}, \mathsf{w},\mu^\otimes, \mu^I, \overline{\mathsf{c}}, \overline{\mathsf{w}})$ is an $\mathbb{N}$-graded monoidal additive bialgebra modality on $\mathsf{REL}$. 
\end{example}

\section{Graded Differential Linear Logic and Coderelictions}\label{sec:gcoder}

In this section, we are finally in a position to properly introduce the graded version of the codereliction. The final rule for $\mathsf{GDiLL}$ is the \textbf{codereliction rule} ($\overline{\mathsf{d}}$): 
\begin{align}
    \begin{prooftree}
\hypo{\Gamma \vdash A}
\infer1[$\overline{\mathsf{d}}$]{\Gamma \vdash \oc_1 A}
\end{prooftree}
\end{align}
From the codereliction rule, one can take a non-linear proof of degree 1 and obtain a linear proof. In fact, the cut-elimination rules tell us that in $\mathsf{GDiLL}$, the linear proofs are precisely the non-linear proofs of degree 1. In categorical terms, this is given by a natural transformation $\overline{\mathsf{d}}: A \to \oc_1 A$, which note is of the opposite type of $\mathsf{d}$. The following is mostly a graded version of \cite[Def 9]{Blute2019}: 

\begin{definition}\label{def:coder} Let $R$ be a f.a.s.\! semiring. An \textbf{$R$-graded monoidal additive bialgebra differential modality} on an additive symmetric monoidal category $\mathcal{L}$ is a tuple $(\oc,\mathsf{p},\mathsf{d},\mathsf{c}, \mathsf{w}, \mu^\otimes, \mu^I, \overline{\mathsf{c}}, \overline{\mathsf{w}}, \overline{\mathsf{d}})$ consisting of an $R$-graded monoidal additive bialgebra modality $(\oc,\mathsf{p},\mathsf{d},\mathsf{c}, \mathsf{w}, \mu^\otimes, \mu^I, \overline{\mathsf{c}}, \overline{\mathsf{w}})$ and a \textbf{codereliction} $\overline{\mathsf{d}}$, which is a natural transformation $\overline{\mathsf{d}}: A \to \oc_1 A$ such that the following equations hold: 
\begin{enumerate}[{\em (i)}]
\item \textbf{Linear Rule:} $\overline{\mathsf{d}};\mathsf{d} = 1$
\item \textbf{Product Rule:} $\overline{\mathsf{d}};\mathsf{c}_{r,s} = \delta_{r,0} \cdot (\overline{\mathsf{w}} \otimes \overline{\mathsf{d}}) + \delta_{0,s} \cdot (\overline{\mathsf{d}} \otimes \overline{\mathsf{w}})$ where $r+s = 1$
\item \textbf{Chain Rule:} $\overline{\mathsf{d}};\mathsf{p}_{1,1} = (\overline{\mathsf{w}} \otimes \overline{\mathsf{d}});(\mathsf{p}_{0,1} \otimes \overline{\mathsf{d}}); \overline{\mathsf{c}}_{0,1}$
\item \textbf{Monoidal Rule:} $(1 \otimes \overline{\mathsf{d}});\mu^\otimes_1 = (\mathsf{d} \otimes 1);\overline{\mathsf{d}}$
  \end{enumerate}
An \textbf{$R$-graded differential linear category} is an additive symmetric monoidal category equipped with a chosen  $R$-graded monoidal additive bialgebra differential modality.  
\end{definition}

Let us quickly compare this with the non-graded version. For starters, there is no longer a constant rule (which was shown to be provable anyways in \cite[Lem 6]{Blute2019}) between $\overline{\mathsf{d}}$ and $\mathsf{w}$, since in general they cannot compose. The explanation for the product rule is the same as for (\ref{dceq}), and recaptures the non-graded version when $R = \mathsf{0}$. While the monoidal rule is provable from the rest in the non-graded case, this is no longer necessarily true in the graded setting. The intuition for the codereliction is that it provides a way to linearize differentiable maps of degree $1$, specifically by taking the derivative and then evaluating at zero. The linear rule says that linearizing a linear map does nothing. The product rule says that linearization of the product of a linear map and a constant map results in the product of the linear map and the constant map. The chain rule tells us how to linearize the composition of differentiable maps that result in one of degree $1$. Lastly, the monoidal rule tells us how the codereliction interacts with the monoidal product. 

We now explain how even in the graded setting, deriving transformations and coderelictions are equivalent. Thus we may claim that there is only one notion of differentiation in categorical models of $\mathsf{GDiLL}$. Syntactically, the construction from one to the other is given as follows: 
\begin{align*}
\begin{prooftree}
\hypo{\Gamma \vdash !_{r}A}
\hypo{\Delta \vdash A}
\infer2[$\partial_{r}$]{\Gamma, \Delta \vdash !_{r+1}A}
\end{prooftree}
:= 
\begin{prooftree}
\hypo{\Gamma \vdash !_{r}A}
\hypo{\Delta \vdash A}
\infer1[$\overline{\mathsf{d}}$]{\Delta \vdash !_{1}A}
\infer2[$\overline{\mathsf{c}}_{r,1}$]{\Gamma, \Delta \vdash !_{r+1}A}
\end{prooftree}
&& 
\begin{prooftree}
\hypo{\Gamma \vdash A}
\infer1[$\overline{\mathsf{d}}$]{\Gamma \vdash !_{1}A}
\end{prooftree}
:= 
\begin{prooftree}
\infer0[$\overline{\mathsf{w}}$]{\vdash !_{0}A}
\hypo{\Gamma \vdash A}
\infer2[$\partial_0$]{\Gamma \vdash !_{1}A}
\end{prooftree}
\end{align*}
We will require an extra axiom for the deriving transformation, which is taking the product rule at zero: 
\begin{enumerate}[{\em (i)}]
\setcounter{enumi}{5}
\item \textbf{Product Rule at $0$:} $\partial_0; \mathsf{c}_{r,s} = \delta_{r,0} \cdot (\mathsf{c}_{0,0} \otimes 1);(1 \otimes \partial_0) +  \delta_{0,s} \cdot (\mathsf{c}_{0,0} \otimes 1); (1 \otimes \sigma); (\partial_0 \otimes 1)$ where $r+s = 1$
  \end{enumerate}
  This is not necessarily captured by the product rule in Def \ref{def:diffcat}, but is still satisfied in the non-graded setting. We may now properly state the graded version of \cite[Thm 4]{Blute2019}. 

\begin{theorem}\label{thm:der-coder} Let $R$ be f.a.s.\! and let $(\oc,\mathsf{p},\mathsf{d},\mathsf{c}, \mathsf{w}, \mu^\otimes, \mu^I, \overline{\mathsf{c}}, \overline{\mathsf{w}})$ be an $R$-graded monoidal additive bialgebra modality on an additive symmetric monoidal category $\mathcal{L}$. Then there is a bijective correspondence between coderelictions and monoidal deriving transformations that also satisfy the product rule at $0$. Explicitly: 
\begin{enumerate}[{\em (i)}]
\item If $\overline{\mathsf{d}}$ is a codereliction, then for every $r \in R$ define: 
\begin{align}
    \partial_r  := \xymatrixcolsep{5pc}\xymatrix{ \oc_r A \otimes A \ar[r]^-{1 \otimes \overline{\mathsf{d}}} & \oc_r A \otimes \oc_1 A \ar[r]^-{\overline{\mathsf{c}}_{r,1}} & \oc_{r+1} A 
  } 
\end{align}
Then $\partial$ is a monoidal deriving transformation which also satisfies the product Rule at $0$.
\item If $\partial$ is a monoidal deriving transformation which also satisfies the product Rule at $0$, then define: 
\begin{align}
    \overline{\mathsf{d}} := \xymatrixcolsep{5pc}\xymatrix{ A \ar[r]^-{\overline{\mathsf{w}} \otimes 1} & \oc_0 A \otimes A \ar[r]^-{\partial_0} & \oc_1 A }
\end{align}   
Then $\overline{\mathsf{d}}$ is a codereliction. 
\end{enumerate}
Furthermore, these constructions are inverses of each other. 
\end{theorem}
\begin{proof} The proof is mostly the same as that of \cite[Thm 4.12]{blute2006differential}. Essentially, one uses the rule for a deriving transformation (resp. codereliction) to prove the rule of the same for a codereliction (resp. deriving transformation). The only tricky one might be the product rule, so let us prove it here. Starting with a deriving transformation $\partial$, we use the product rule at $0$ and (\ref{bialgeq}), to prove the product rule for the codereliction: 
\begin{gather*}
\overline{\mathsf{d}};\mathsf{c}_{r,s} = (\overline{\mathsf{w}} \otimes 1); \partial_0; \mathsf{c}_{r,s}= \delta_{r,0} \cdot  (\overline{\mathsf{w}} \otimes 1);(\mathsf{c}_{0,0} \otimes 1);(1 \otimes \partial_0) + \delta_{0,s} \cdot (\overline{\mathsf{w}} \otimes 1);(\mathsf{c}_{0,0} \otimes 1); (1 \otimes \sigma); (\partial_0 \otimes 1) \\
= \delta_{r,0} \cdot  (\overline{\mathsf{w}} \otimes \overline{\mathsf{w}} \otimes 1);(1 \otimes \partial_0) + \delta_{0,s} \cdot (\overline{\mathsf{w}} \otimes 1 \otimes \overline{\mathsf{w}}); (\partial_0 \otimes 1) = \delta_{r,0} \cdot (\overline{\mathsf{w}} \otimes \overline{\mathsf{d}}) + \delta_{0,s} \cdot (\overline{\mathsf{d}} \otimes \overline{\mathsf{w}})
\end{gather*}
 Conversely, starting from a codereliction, we use the axiom of a f.a.s.\! semiring which says that $c+d= 1$ implies that $c=1$ and $d=0$, $c=0$ and $d=1$. 
\begin{gather*}
\partial_{r+s+1}; \mathsf{c}_{r+1, s+1} = (1 \otimes \overline{\mathsf{d}}); \overline{\mathsf{c}}_{r+s+1,1}; \mathsf{c}_{r+1, s+1} = \sum\limits_{\substack{ a+b = r+s+1 \quad c+d = 1 \\ a+c = r+1 \quad b+d = s+1}} \!\!\!\!\!\! (1 \otimes \overline{\mathsf{d}});(\mathsf{c}_{a,b} \otimes \mathsf{c}_{c,d}); (1 \otimes \sigma \otimes 1); (\overline{\mathsf{c}}_{a,c} \otimes \overline{\mathsf{c}}_{b,d}) \\
=  (1 \otimes \overline{\mathsf{d}});(\mathsf{c}_{r+1,s} \otimes \mathsf{c}_{0,1}); (1 \otimes \sigma \otimes 1); (\overline{\mathsf{c}}_{r+1,0} \otimes \overline{\mathsf{c}}_{s,1}) + (1 \otimes \overline{\mathsf{d}});(\mathsf{c}_{r,s+1} \otimes \mathsf{c}_{1,0}); (1 \otimes \sigma \otimes 1); \overline{\mathsf{c}}_{r,1} \otimes \overline{\mathsf{c}}_{s+1,0} \\
= (\mathsf{c}_{r+1,s} \otimes \overline{\mathsf{w}} \otimes \overline{\mathsf{d}}); (1 \otimes \sigma \otimes 1); (\overline{\mathsf{c}}_{r+1,0} \otimes \overline{\mathsf{c}}_{s,1}) +  (\mathsf{c}_{r,s+1} \otimes \overline{\mathsf{d}} \otimes \overline{\mathsf{w}}); (1 \otimes \sigma \otimes 1); (\overline{\mathsf{c}}_{r,1} \otimes \overline{\mathsf{c}}_{s+1,0})  \\
=  (\mathsf{c}_{r+1,s} \otimes \overline{\mathsf{d}}); (1 \otimes \overline{\mathsf{c}}_{s,1}) +  (\mathsf{c}_{r,s+1} \otimes \overline{\mathsf{d}}); (1 \otimes \sigma); (\overline{\mathsf{c}}_{r,1} \otimes 1)\\= (\mathsf{c}_{r+1,s} \otimes 1);(1 \otimes \partial_s) +  (\mathsf{c}_{r,s+1} \otimes 1); (1 \otimes \sigma); (\partial_r \otimes 1)  
\end{gather*}
It is also straightforward to see that these constructions are indeed inverses of each other. 
\end{proof}

\begin{example} A $\mathsf{0}$-graded differential linear category is precisely a non-graded differential linear category. 
\end{example}

\begin{example}\label{ex:RELcoder} Define the following relation: 
\begin{align*}
    \overline{\mathsf{d}} := \lbrace (x,\overline{\mathsf{d}}_x) ~\vert~ \forall x \in X \rbrace \subset X \times \oc_1 X 
\end{align*}
Then $(\oc,\mathsf{p},\mathsf{d},\mathsf{c}, \mathsf{w},\mu^\otimes, \mu^I, \overline{\mathsf{c}}, \overline{\mathsf{w}}, \overline{\mathsf{d}})$ is an $\mathbb{N}$-graded monoidal additive bialgebra differential modality on $\mathsf{REL}$ (which again follows from the results in Section \ref{sec:sympowers}), and the induced deriving transformation is precisely that of Ex \ref{ex:RELd}. Furthermore, note that if $X$ is a finite set, then $\oc_n X$ is also a finite set for any $n \in \mathbb{N}$. Therefore, this $\mathbb{N}$-graded monoidal additive bialgebra differential modality also restricts to the subcategory $\mathsf{FREL}$ of finite sets and relations. Thus $\mathsf{REL}$ and $\mathsf{FREL}$ are both $\mathbb{N}$-graded differential linear categories. 
\end{example}

\section{Symmetric Powers}\label{sec:sympowers}

In this section, we explain how symmetric powers always provide an example of an $\mathbb{N}$-graded differential linear category. Symmetric powers are an important concept in both Linear Logic and classical algebra, though they refer to the dual notion of each other. In Linear Logic, the term symmetric power refers to the equalizer of all permutations on the monoidal product, and symmetric powers are used to build free exponential modalities \cite[Sec B]{ong2017quantitative}. On the other hand, in algebra, the term symmetric power refers to the coequalizers of all permutations on the tensor product, and here symmetric powers are instead used to build the free symmetric algebra. To avoid confusion, we will refer to them as (co)equalizer symmetric powers. Though we will also discuss the case when both notions coincide. 

For each $n \in \mathbb{N}$, let $\Sigma(n)$ be the symmetric group of all $n!$ permutations. For every permutation $\tau \in \Sigma(n)$ we obtain a natural isomorphism $\tau: A_1 \otimes \hdots \otimes A_n \to A_{\tau^{-1}(1)} \otimes \hdots \otimes A_{\tau^{-1}(n)}$. A symmetric monoidal category $\mathcal{L}$ is said to have \textbf{equalizer symmetric powers} \cite[Def 12]{ong2017quantitative} if for all objects $A \in \mathcal{L}$ and all $n \in \mathbb{N}$ there exists an object $A^{[n]}$ and a map $\epsilon_n: A^{[n]} \to A^{\otimes^n}$ which is the equalizer of all $n!$ permutations on $A^{\otimes^n}$ and such that $\otimes$ preserves these equalizers, that is, for all $n,m \in \mathbb{N}$, $\epsilon_n \otimes \epsilon_m: A^{[n]} \otimes A^{[m]} \to A^{\otimes^n} \otimes A^{\otimes^m}$ is the equalizer for all $\tau \otimes \tau^\prime$ where $\tau \in \Sigma(n)$ and $\tau^\prime \in \Sigma(n)$. By convention, $A^{[0]}=I$ with $\epsilon_I = 1_I$ and $A^{[1]} = A$ with $\epsilon_1 = 1_A$, while for $I$ we use the convention that $I^{[n]} = I$ and $\epsilon_n = 1_I$. 

Equalizer symmetric powers always produce an $\mathbb{N}$-graded monoidal coalgebra modality where: $\oc_n A = A^{[n]}$, $\mathsf{d} = \epsilon_1 = 1_A$, $\mathsf{w} = \epsilon_0 = 1_I$, and $\mu^I_n = 1_I$, while $\mathsf{p}_{n,m}$, $\mathsf{c}_{n,m}$, and $\mu^\otimes_n$ are defined using the equalizer universal property and so are defined as the unique maps such that:
\begin{align}
 \mathsf{c}_{n,m}; (\epsilon_n \otimes \epsilon_m) = \epsilon_{n+m} &&  \mathsf{p}_{n,m};\epsilon_n;\epsilon_m^{\otimes^n} = \epsilon_{nm} && \mu^\otimes_n;\epsilon_n = (\epsilon_n \otimes \epsilon_n); \theta
\end{align}
where $\theta: A^{\otimes^n} \otimes B^{\otimes^n} \to (A \otimes B)^{\otimes^n}$ is the permutation that rearranges the $A$'s and $B$'s but keeping their order. Intuitively, $\mathsf{c}_{n,m}$ is the equalization of all possible reorganizations $A^{\otimes^{n+m}} \to A^{\otimes^n} \otimes A^{\otimes^m}$, while $\mathsf{p}_{n,m}$ is the equalization of all possible reorganizations $A^{\otimes^{nm}} \to {A^{\otimes^m}}^{\otimes^n}$, and $\mu^\otimes_n$ is the equalization of all possible reorganizations $A^{\otimes^n} \otimes B^{\otimes^n} \to (A \otimes B)^{\otimes^n}$. 

Now suppose that we are also in an additive symmetric monoidal category. Then equalizer symmetric powers always give an $\mathbb{N}$-graded monoidal additive bialgebra differential modality where $\overline{\mathsf{d}} = 1_A$, $\overline{\mathsf{w}} = 1_I$, and $\overline{\mathsf{c}}_{n,m}$ is defined using the equalizer universal property, so the unique map such that:
\begin{align}
    \overline{\mathsf{c}}_{n,m}; \epsilon_{n+m} = \binom{n+m}{n} \cdot (\epsilon_n \otimes \epsilon_m)
\end{align}
The binomial coefficient for $\overline{\mathsf{c}}$ is necessary for the graded bimonoid identity. If we also have finite biproducts, then by Thm \ref{thm:Seely} we also obtain the graded Seely isomorphism, so: $(A \oplus B)^{[n]} \cong \bigoplus^n_{k=0} A^{[k]} \otimes B^{[n-k]}$. 

Now the induced deriving transformation will simply be $\partial_n = \overline{\mathsf{c}}_{n,1}$, which is the unique map such that:
\begin{align}
    \partial_n; \epsilon_{n+1} = (n+1) \cdot (\epsilon_n \otimes 1_A )
\end{align} 
In other words, $\partial_n$ is given by equalizing all possible $n+1$ ways of inserting an extra $A$ into $A^{\otimes^n}$. We will explain below how this (in the dual case) corresponds precisely to differentiating polynomials.
 
\begin{proposition}\label{prop:sympowers} An additive symmetric monoidal category with (co)equalizer symmetric powers is an $\mathbb{N}$-graded (co)differential linear category. 
\end{proposition}

\begin{example}\label{ex:RELsym} $\mathsf{REL}$ has equalizer symmetric powers, and we have already in fact defined them. Indeed, we may set $X^{[n]} = \oc_n X$, where $\oc_n X$ is defined as in Ex \ref{ex:REL!} and define the relation:
\begin{align*}
    \epsilon_n = \lbrace (f, (x_1, \hdots, x_n)) ~\vert~ x_i \in X, f \in \oc_n X \text{ s.t. } f(x)=1 \text{ if } x=x_i \text{ and } f(x) =0 \text{ o.w. } \rbrace \subseteq \oc_n X \times X^{\times^n}
\end{align*}
The $\mathbb{N}$-graded differential linear structure induced by the equalizer symmetric powers on $\mathsf{REL}$ is exactly the one defined in the above examples. 
\end{example}

\begin{example}\label{ex:mod} Let $k$ be a commutative semiring and let $k\text{-}\mathsf{MOD}$ be the category of modules over $k$ and $k$-linear maps between them. $k\text{-}\mathsf{MOD}$ is an additive symmetric monoidal category with the standard algebraic tensor product and $k$-linear enrichment. It is well-known that $k\text{-}\mathsf{MOD}$ is both complete and cocomplete, which in particular implies that $k\text{-}\mathsf{MOD}$ has all equalizers and coequalizers. However, the tensor product $\otimes$ does not necessarily preserve equalizers, and so $k\text{-}\mathsf{MOD}$ may not always have equalizer symmetric powers. On the other hand, $\otimes$ preserves all colimits and so $k\text{-}\mathsf{MOD}$ has coequalizer symmetric powers and is, therefore, an $\mathbb{N}$-graded codifferential category. The coequalizer symmetric powers are well-known since they are the symmetrized tensor spaces used in the construction of the symmetric algebra. For a $k$-module $M$, the $n$-th coequalizer symmetric power $S^{n}M$ is defined as the quotient of $M^{\otimes^n}$ by the action of all $n!$ permutations. So an arbitrary element of $S^n M$ can be expressed in terms of finite sums of \textbf{pure symmetric tensors} $x_{1} \otimes_{s} ... \otimes_{s} x_{n}$ for all $x_i \in M$, and these satisfy that for all $n!$ permutations $\tau$: $x_{\tau(1)} \otimes_{s} ... \otimes_{s} x_{\tau(n)} = x_{1} \otimes_{s} ... \otimes_{s} x_{n}$. So $\epsilon_n: M^{\otimes^n} \to S^n(M)$ is defined as:
\[\epsilon_n(x_1 \otimes \hdots \otimes x_n) = x_{1} \otimes_{s} ... \otimes_{s} x_{n}\] 
The resulting graded codifferential structure is given as follows (if the maps look backwards, recall we are in the dual setting): 
\begin{gather*}
\mathsf{d}(x) = x \quad \quad \quad \mathsf{p}_{m,n} \left( (x_{1} \otimes_{s} ... \otimes_{s} x_{n}) \otimes_{s} ... \otimes_{s} (x_{nm-n+1} \otimes_{s} ... \otimes_{s} x_{nm}) \right) = x_{1} \otimes_{s} ... \otimes_{s} x_{nm} \\
 \mathsf{w}(1) =1 \quad \quad \quad  \mathsf{c}_{n,m}\left((x_{1} \otimes_{s} ... \otimes_{s} x_{n}) \otimes (x_{n+1} \otimes_{s} ... \otimes_{s} x_{n+m}) \right) = x_{1} \otimes_{s} ... \otimes_{s} x_{n+m} \\
\mu^I_n(1) = 1 \quad \quad \mu^\otimes_n\left((x_{1} \otimes y_1) \otimes_{s} ... \otimes_{s} (x_{n} \otimes y_n) \right) = (x_{1} \otimes_{s} ... \otimes_{s} x_{n}) \otimes (y_{1} \otimes_{s} ... \otimes_{s} y_{n}) \\ 
\overline{\mathsf{w}}(1) = 1 \quad \quad \overline{\mathsf{c}}_{n,m}\left(x_{1} \otimes_{s} ... \otimes_{s} x_{n+m}\right) = \sum\limits_{\substack{1 \leq i_j \leq n+m \\ i_1 < \hdots < i_n \\ i_{n+1}  < ... < i_{n+m}}} \left( x_{i_1}\otimes_s \hdots \otimes_{s} x_{i_n} \right) \otimes \left(x_{i_{n+1}} \otimes_s \hdots \otimes_{s} x_{i_{n+m}} \right) \\ 
\overline{\mathsf{d}}(x) = x \quad \quad \quad \partial_{n}(x_{1} \otimes_{s} ... \otimes_{s} x_{n+1}) = \sum^{n}_{i=0} \left(x_{1} \otimes_{s} ... \otimes_{s} x_{i-1} \otimes x_{i+1} \otimes_{s} ... \otimes_{s} x_{n+1} \right) \otimes x_{i}
\end{gather*}
Now if $M$ is a free $k$-module with basis $X= \lbrace x_1, x_2, \hdots, x_m \rbrace$, then $S^{n}M$ is isomorphic to $k_n[x_1, \hdots, x_m]$, the submodule of the polynomial algebra $k[x_1, \hdots, x_m]$ generated by monomials of degree $n$. From this point of view, $\mathsf{p}$ corresponds to polynomial composition, $\mathsf{d}$ gives the monomials of degree $1$, $\mathsf{c}$ is given by polynomial multiplication, $\mathsf{w}$ gives the constant polynomials, and lastly $\partial$ correspond to differentiating polynomials. Indeed, if $p(\vec x) \in k_n[x_1, \hdots, x_m]$, then:
\[\partial_n(p(\vec x)) = \sum^m_{i=1} \frac{\partial p(\vec x)}{\partial x_m} \otimes x_i\] 
Thus by Prop \ref{prop:sympowers}, $k\text{-}\mathsf{MOD}$ is an $\mathbb{N}$-graded codifferential linear category, and therefore $k\text{-}\mathsf{MOD}^{op}$ is an $\mathbb{N}$-graded differential linear category. Furthermore, if $M$ is finite-dimensional, then $S^n M$ are all finite-dimensional as well. Therefore, for the subcategory of finite-dimensional $k$-modules, $k\text{-}\mathsf{FMOD}$, its opposite category $k\text{-}\mathsf{FMOD}^{op}$ is also an $\mathbb{N}$-graded differential linear category. In particular when $k=\mathbb{C}$, the field of complex numbers, $\mathbb{C}\text{-}\mathsf{FMOD} \cong \mathsf{FHILB}$. So $\mathsf{FHILB}^{op}$ is also an $\mathbb{N}$-graded differential linear category, solving the problem of \cite{lemay2020fhilb}. 
\end{example}

Now there are many examples where the $n$-th equalizer symmetric power $A^{[n]}$ is also the $n$-th coequalizer symmetric power, as is the case in the category of vector spaces over a field of characteristic $0$ \cite[Chap III. Sec 6.3]{bourbaki1998algebra}. A particular instance of when this is the case is when a certain idempotent splits. First, let us now work in $\mathbb{Q}_{\geq 0}$-additive symmetric monoidal categories, which are additive symmetric monoidal categories whose homsets are also $\mathbb{Q}_{\geq 0}$-modules, so we may scalar multiply by positive rationals $\frac{p}{q} \in \mathbb{Q}_{\geq 0}$. In this case, for every object $A$ and every $n \in \mathbb{N}$, define the natural transformation $\rho(n): A^{\otimes^n} \to A^{\otimes^n}$ as the sum of all $n!$ permutations on $A^{\otimes n}$ divided by a factor of $n!$, so:
\[\rho(n) := \sum_{\tau \in \Sigma(n)} \frac{1}{n!} \cdot \tau\]
Observe that $\rho(n)$ is an idempotent. Indeed, first note that since $\Sigma(n)$ is a finite group, $\rho(n)$ both equalizes and coequalizes all permutations, that is, for all permutation $\tau^\prime \in \Sigma(n)$, we have that ${\rho(n); \tau^\prime = \rho(n) = \tau^\prime; \rho(n)}$. Then since $\Sigma(n)$ is of size $n!$, we have that: 
\[ \rho(n); \rho(n) = \left( \sum_{\tau \in \Sigma(n)} \frac{1}{n!} \cdot \tau \right); \rho(n) = \sum_{\tau \in \Sigma(n)} \frac{1}{n!} \cdot \tau; \rho(n) = \sum_{\tau \in \Sigma(n)} \frac{1}{n!} \cdot \rho(n) = \frac{n!}{n!} \cdot \rho(n) = \rho(n) \]
So $\rho(n)$ is indeed an idempotent. Now suppose that $\rho(n)$ is, in fact, a split idempotent, which means there is an object $A^{[n]}$ and maps $\epsilon_n: A^{[n]} \to A^{\otimes^n}$ and $\varsigma_n: A^{\otimes^n} \to A^{[n]}$, such that: 
\begin{align*}
    \epsilon_n; \varsigma_n = 1_{\oc_n A } && \varsigma_n; \epsilon_n = \rho(n)
\end{align*}
For split idempotents, it is well-known that $A^{[n]}$ is both the equalizer and coequalizer of the identity $1$ and the idempotent $\rho(n)$, and these are absolute (co)limits, so in particular, $\otimes$ preserves these (co)equalizers.  Furthermore, since $\rho(n)$ also equalizes and coequalizes all permutations, it follows that $A^{[n]}$ is indeed the $n$-th (co)equalizer symmetric power. Thus a $\mathbb{Q}_{\geq 0}$-additive symmetric monoidal category where $\rho(n)$ splits for all $n$ has equalizer symmetric powers $A^{[n]}$ which are also coequalizer symmetric powers, and therefore by Prop \ref{prop:sympowers} is both an $\mathbb{N}$-graded differential linear category and an  $\mathbb{N}$-graded codifferential linear category. Using the splitting of $\rho(n)$, we can define all the necessary structure maps for the monoidal additive bialgebra differential modality as follows: 
  \begin{equation}\begin{gathered}\label{eq:diff-split}
\mathsf{p}_{n,m} = \epsilon_{nm}; \varsigma^{\otimes^n}_m ; \varsigma_n \quad  \quad  \quad  \mathsf{d} = 1_A  \quad  \quad  \quad \mathsf{c}_{n,m} = \epsilon_{n+m};(\varsigma_n \otimes \varsigma_m) \quad \quad  \quad  \mathsf{w}=1_I  \\
\mu^\otimes_n = (\epsilon_n \otimes \epsilon_n); \cong ; \varsigma_{n} \quad\quad  \quad   \mu^I_n = 1_I  \quad  \quad  \quad \overline{\mathsf{c}}_{n,m} = \binom{n+m}{n} \cdot (\epsilon_n \otimes \epsilon_m); \varsigma_{n+m} \quad\quad  \quad   \overline{\mathsf{w}} = 1_I \\
\overline{\mathsf{d}}=1_A \quad \quad \quad \partial_n = (n+1) \cdot (\epsilon_n \otimes 1_A); \varsigma_{n+1}
\end{gathered}\end{equation}
By dualizing these formulas, we also obtain the necessary structure for an $\mathbb{N}$-graded codifferential category. 

\begin{example} $\mathsf{REL}$ is a $\mathbb{Q}_{\geq 0}$-additive symmetric monoidal category since addition is idempotent in $\mathsf{REL}$, that is, $S + S = S$. In particular, the idempotent relation $\rho(n) \subseteq X^{\times^n} \times X^{\times^n}$ splits through $\oc_n X$ via the relation $\epsilon_n \subseteq  \oc_n X \times X^{\times^n}$, as defined in Ex \ref{ex:RELsym}, and the relation: 
\begin{align*}
    \varsigma_n = \lbrace ((x_1, \hdots, x_n),\sum^n_{i=1} \overline{\mathsf{d}}_{x_i}) ~\vert~ x_i \in X \rbrace \subseteq X^{\times^n} \times \oc_n X 
\end{align*}
Applying the constructions of (\ref{eq:diff-split}) results precisely in the differential structure described above. Applying the dual of (\ref{eq:diff-split}) also makes $\mathsf{REL}$ an $\mathbb{N}$-graded codifferential category -- which also follows from the fact that $\mathsf{REL}$ is self-dual. 
\end{example}

\begin{example} Let $k$ be a commutative semiring that is also a $\mathbb{Q}_{\geq 0}$-algebra. In particular this implies that $k\text{-}\mathsf{MOD}$ is a $\mathbb{Q}_{\geq 0}$-additive symmetric monoidal category. It is well known that all idempotents split in $k\text{-}\mathsf{MOD}$. In particular, the idempotent $\rho(n): M^{\otimes^n} \to M^{\otimes^n}$ can split through $S^n M$ via the maps $\epsilon_n: S^n M \to M^{\otimes^n}$ and $\varsigma_n: M^{\otimes^n} \to S^n M$ respectively defined as follows \cite[Chap III. Sec 6.3]{bourbaki1998algebra}: 
\begin{align*}
    \epsilon_n\left( x_1 \otimes_s \hdots \otimes_s x_n \right) = \frac{1}{n!} \cdot \sum\limits_{\tau \in \Sigma(n)} x_{\tau(1)} \otimes \hdots \otimes x_{\tau(n)} && \varsigma_n\left(x_1 \otimes \hdots \otimes x_n  \right) = x_1 \otimes_s \hdots \otimes_s x_n 
\end{align*}
Applying the dual constructions of (\ref{eq:diff-split}) results precisely in the codifferential structure given in Ex \ref{ex:mod}. Interestingly, we could also apply (\ref{eq:diff-split}) directly to $k\text{-}\mathsf{MOD}$ instead to obtain an $\mathbb{N}$-graded differential structure, where in particular the deriving transformation is given by: 
\begin{align*}
     \partial_{n}\left( (x_{1} \otimes_{s} ... \otimes_{s}  x_{n}) \otimes x \right) = (n+1) \cdot x_{1} \otimes_{s} ... \otimes_{s}  x_{n} \otimes_s x 
\end{align*}
So $k\text{-}\mathsf{MOD}$ is also an $\mathbb{N}$-graded differential linear category. Note that the resulting graded bimonoid is not the same that the one in Ex \ref{ex:mod}, though they are  isomorphic (see \cite{vienney2023string} for more details). Nevertheless, this structure of differential category as a whole is somewhat mysterious.
\end{example}

\section{Conclusion: Future Work}

This paper introduced graded differential categories, a new branch in the theory of differential categories. There remain many intriguing questions to answer about graded differential categories and many possible paths for interesting future work. Let us conclude this paper by briefly describing some.
\begin{enumerate}[{\em (i)}]
\item An obvious objective is always to construct and study new examples of graded differential categories and models of $\mathsf{GDiLL}$.
\item We could study more carefully the relation between syntax and semantics, in particular the cut-elimination for the different sequent calculus presented here. It would also be interesting to understand the relations of graded differential categories with the differential linear logic indexed by differential operators of \cite{breuvart2023unifying}, and with the Taylor expansion of differential linear logic, notably as treated in the recent work \cite{kerjean2023taylor}. 
\item Also, in most models of $\mathsf{GLL}$, one considers the grading to be over an ordered semiring. We could investigate what order adds to the story of differentiation. 
\item It is natural to ask whether it is possible to take a graded differential category and somehow glue together the graded exponential $\oc_r$ to obtain a differential category. Or conversely, when is it possible to start from a non-graded differential category and somehow extract a graded differential category.
\item Another question is whether it is possible to build a Cartesian differential category \cite{blute2009cartesian} from a graded differential category, or possibly some yet-to-be-defined graded version of a Cartesian differential category. 
\end{enumerate}
So in conclusion, this paper is only the start of the story of graded differential categories -- there are many interesting questions to investigate and exciting future work still to be done. 

\section*{Acknowledgement} 
The authors would like to thank Rick Blute, Flavien Breuvart, Cole Comfort, Sacha Ikonicoff, Shin-ya Katsumata, Marie Kerjean, Paul-André Melliès, Simon Mirwasser, Phil Scott, Christine Tasson and Lionel Vaux for many useful discussions and their support of this research project. The authors would also like to thank Flavien Breuvart for pointing out an error in an example, which we have now fixed. 


\bibliographystyle{./entics}     
\bibliography{MFPSbib}   

\begin{thebibliography}{10}
\expandafter\ifx\csname url\endcsname\relax
  \def\url#1{\texttt{#1}}\fi
\expandafter\ifx\csname urlprefix\endcsname\relax\def\urlprefix{URL }\fi
\newcommand{\enquote}[1]{``#1''}

\bibitem{ardizzoni2012associated}
Ardizzoni, A. and C.~Menini, \emph{Associated graded algebras and coalgebras},
  Communications in Algebra \textbf{40} (2012), pp.~862--896.
\newline\urlprefix\url{https://doi.org/10.48550/arXiv.0704.2106}

\bibitem{bierman1995categorical}
Bierman, G.~M., \emph{What is a categorical model of intuitionistic linear
  logic?}, in: \emph{International Conference on Typed Lambda Calculi and
  Applications}, Springer, 1995, pp. 78--93.
\newline\urlprefix\url{https://doi.org/10.1007/BFb0014046}

\bibitem{Blute2019}
Blute, R.~F., J.~R.~B. Cockett, J.-S.~P. Lemay and R.~A.~G. Seely,
  \emph{{Differential Categories Revisited}}, Applied Categorical Structures
  \textbf{28} (2020).
\newline\urlprefix\url{https://doi.org/10.48550/arXiv.1806.04804}

\bibitem{blute2006differential}
Blute, R.~F., J.~R.~B. Cockett and R.~A.~G. Seely, \emph{Differential
  categories}, Mathematical Structures in Computer Science \textbf{16} (2006),
  pp.~1049--1083.
\newline\urlprefix\url{https://doi.org/10.1017/S0960129506005676}

\bibitem{blute2009cartesian}
Blute, R.~F., J.~R.~B. Cockett and R.~A.~G. Seely, \emph{Cartesian differential
  categories}, Theory and Applications of Categories \textbf{22} (2009),
  pp.~622--672.
\newline\urlprefix\url{http://www.tac.mta.ca/tac/volumes/22/23/22-23abs.html}

\bibitem{blute2015cartesian}
Blute, R.~F., J.~R.~B. Cockett and R.~A.~G. Seely, \emph{Cartesian differential
  storage categories}, Theory and Applications of Categories \textbf{30}
  (2015), pp.~620--686.
\newline\urlprefix\url{http://www.tac.mta.ca/tac/volumes/30/18/30-18abs.html}

\bibitem{bourbaki1998algebra}
Bourbaki, N., \enquote{Algebra I,} Springer, 1998.
\newline\urlprefix\url{https://www.google.ca/books/edition/Algebra_I/STS9aZ6F204C?hl=fr&gbpv=0}

\bibitem{breuvart2023unifying}
Breuvart, F., M.~Kerjean and S.~Mirwasser, \emph{{Unifying Graded Linear Logic
  and Differential Operators}}, in: M.~Gaboardi and F.~van Raamsdonk, editors,
  \emph{8th International Conference on Formal Structures for Computation and
  Deduction (FSCD 2023)},  Leibniz International Proceedings in Informatics
  (LIPIcs)  \textbf{260} (2023), pp. 21:1--21:21.
\newline\urlprefix\url{https://drops.dagstuhl.de/opus/volltexte/2023/18005}

\bibitem{breuvart2015modelling}
Breuvart, F. and M.~Pagani, \emph{{Modelling Coeffects in the Relational
  Semantics of Linear Logic}}, in: S.~Kreutzer, editor, \emph{24th EACSL Annual
  Conference on Computer Science Logic (CSL 2015)},  Leibniz International
  Proceedings in Informatics (LIPIcs)  \textbf{41} (2015), pp. 567--581.
\newline\urlprefix\url{https://doi.org/10.4230/LIPIcs.CSL.2015.567}

\bibitem{brunel2014core}
Brunel, A., M.~Gaboardi, D.~Mazza and S.~Zdancewic, \emph{{A Core Quantitative
  Coeffect Calculus}}, in: Z.~Shao, editor, \emph{Programming Languages and
  Systems} (2014), pp. 351--370.
\newline\urlprefix\url{https://doi.org/10.1007/978-3-642-54833-8_19}

\bibitem{carraro2010exponentials}
Carraro, A., T.~Ehrhard and A.~Salibra, \emph{Exponentials with infinite
  multiplicities}, in: \emph{Computer Science Logic: 24th International
  Workshop, CSL 2010, 19th Annual Conference of the EACSL, Brno, Czech
  Republic, August 23-27, 2010. Proceedings 24}, Springer, 2010, pp. 170--184.
\newline\urlprefix\url{https://doi.org/10.1007/978-3-642-15205-4_16}

\bibitem{cockett2017there}
Cockett, J. R.~B. and J.-S.~P. Lemay, \emph{{There Is Only One Notion of
  Differentiation}}, in: D.~Miller, editor, \emph{2nd International Conference
  on Formal Structures for Computation and Deduction (FSCD 2017)},  Leibniz
  International Proceedings in Informatics (LIPIcs)  \textbf{84} (2017), pp.
  13:1--13:21.
\newline\urlprefix\url{https://doi.org/10.4230/LIPIcs.FSCD.2017.13}

\bibitem{ehrhard2017introduction}
Ehrhard, T., \emph{{An introduction to Differential Linear Logic: proof-nets,
  models and antiderivatives}}, Mathematical Structures in Computer Science
  (2017).
\newline\urlprefix\url{https://doi.org/10.48550/arXiv.1606.01642}

\bibitem{ehrhard2006differential}
Ehrhard, T. and L.~Regnier, \emph{Differential interaction nets}, Theoretical
  Computer Science \textbf{364} (2006).
\newline\urlprefix\url{https://doi.org/10.1016/j.tcs.2006.08.003}

\bibitem{fiore2007differential}
Fiore, M.~P., \emph{{Differential Structure in Models of Multiplicative
  Biadditive Intuitionistic Linear Logic}}, in: \emph{Typed Lambda Calculi and
  Applications} (2007), pp. 163--177.
\newline\urlprefix\url{https://doi.org/10.1007/978-3-540-73228-0_13}

\bibitem{gaboardi2016combining}
Gaboardi, M., S.~Katsumata, D.~Orchard, F.~Breuvart and T.~Uustalu,
  \emph{Combining effects and coeffects via grading}, ACM SIGPLAN Notices
  \textbf{51} (2016), pp.~476--489.
\newline\urlprefix\url{https://doi.org/10.1145/2951913.2951939}

\bibitem{girard_linear_1987}
Girard, J.-Y., \emph{Linear logic}, Theoret. Comput. Sci. \textbf{50} (1987).
\newline\urlprefix\url{https://doi.org/10.1016/0304-3975(87)90045-4}

\bibitem{girard1992bounded}
Girard, J.-Y., A.~Scedrov and P.~J. Scott, \emph{Bounded linear logic: a
  modular approach to polynomial-time computability}, Theoretical computer
  science \textbf{97} (1992), pp.~1--66.
\newline\urlprefix\url{https://doi.org/10.1016/0304-3975(92)90386-T}

\bibitem{katsumata2018double}
Katsumata, S., \emph{A double category theoretic analysis of graded linear
  exponential comonads}, in: \emph{Foundations of Software Science and
  Computation Structures: 21st International Conference, FOSSACS 2018, Held as
  Part of the European Joint Conferences on Theory and Practice of Software,
  ETAPS 2018, Thessaloniki, Greece, April 14--20, 2018. Proceedings 21},
  Springer, 2018, pp. 110--127.
\newline\urlprefix\url{https://doi.org/10.1007/978-3-319-89366-2_6}

\bibitem{kerjean2023taylor}
Kerjean, M. and J.-S.~P. Lemay, \emph{{Taylor Expansion as a Monad in Models of
  {DiLL}}}, in: \emph{2023 38th Annual ACM/IEEE Symposium on Logic in Computer
  Science (LICS)} (2023), pp. 1--13.
\newline\urlprefix\url{https://doi.ieeecomputersociety.org/10.1109/LICS56636.2023.10175753}

\bibitem{lemay2020fhilb}
Lemay, J.-S.~P., \emph{{Why FHilb is Not an Interesting (Co)Differential
  Category}}, in: \emph{Proceedings 16th International Conference on Quantum
  Physics and Logic, {QPL} 2019, Chapman University, Orange, CA, USA, June
  10-14, 2019},  {EPTCS}  \textbf{318}, 2019, pp. 13--26.
\newline\urlprefix\url{https://doi.org/10.4204/EPTCS.318.2}

\bibitem{arxiv2023version}
Lemay, J.-S.~P. and J.-B. Vienney, \emph{{Graded Differential Categories and
  Graded Differential Linear Logic (with Appendix)}}, Online preprint version
  with appendix.  (2024).
\newline\urlprefix\url{https://github.com/jsplemay/MFPS2023/blob/1c73562b5af63678b66fc8d1ddb7621c6143c3ff/MFPS_2023.pdf}

\bibitem{mellies_categorical_2008}
Melliès, P.-A., \emph{{Categorical Semantics of Linear Logic}}, Société
  Mathématique de France  (2008).
\newline\urlprefix\url{https://smf.emath.fr/publications/semantique-categorielle-de-la-logique-lineaire}

\bibitem{ong2017quantitative}
Ong, C.-H.~L., \emph{Quantitative semantics of the lambda calculus: Some
  generalisations of the relational model}, in: \emph{2017 32nd Annual ACM/IEEE
  Symposium on Logic in Computer Science (LICS)}, IEEE, 2017, pp. 1--12.
\newline\urlprefix\url{https://doi.org/10.1109/LICS.2017.8005064}

\bibitem{sweedler1969hopf}
Sweedler, M.~E., \enquote{Hopf Algebras,} Mathematics lecture note series, W.
  A. Benjamin, 1969.
\newline\urlprefix\url{https://books.google.ca/books?id=8FnvAAAAMAAJ}

\bibitem{vienney2023string}
Vienney, J.-B., \emph{{String diagrams for symmetric powers I: In symmetric
  monoidal $\mathbb{Q}_{\ge 0}$-linear categories}}, arXiv preprint
  arXiv:2308.02094  (2023).
\newline\urlprefix\url{https://doi.org/10.48550/arXiv.2308.02094}

\end{thebibliography}

\end{document}